\documentclass[11pt, letterpaper]{article}

\RequirePackage[OT1]{fontenc}
\RequirePackage{graphicx}
\RequirePackage{amsthm}
\RequirePackage{amssymb}
\RequirePackage[cmex10]{amsmath}
\RequirePackage{natbib}
\RequirePackage[colorlinks,citecolor=blue,urlcolor=blue]{hyperref}
\RequirePackage{hypernat}
\RequirePackage{enumerate}
\RequirePackage{color}

\usepackage[margin=3cm]{geometry}
\usepackage{verbatim}

\providecommand{\keywords}[1]
{
  \small	
  \textbf{\textit{Keywords---}} #1
}

\numberwithin{equation}{section}
\theoremstyle{plain}

\newtheorem{theorem}{Theorem}[section]
\newtheorem{proposition}[theorem]{Proposition}
\newtheorem{lemma}[theorem]{Lemma}
\newtheorem{definition}[theorem]{Definition}
\newtheorem{example}[theorem]{Example}

\newtheorem{notation}[theorem]{Notation}

\newcommand{\R}{\mathbb{R}}
\newcommand{\N}{\mathbb{N}}
\newcommand{\D}{{\mathcal D}}

\newcommand{\ca}[1]{\mathop{\mathrm{Ca}{(#1)}}}
\newcommand{\E}{{\mathrm E}}
\renewcommand{\strut}{\relax}

\title{Existence and Uniqueness of Solutions to the Stochastic Bellman Equation with Unbounded Shock 
}
\author{{Juan Pablo} {Rinc{\'o}n-Zapatero}$^{1}$ \\
\small $^{1}$ Universidad Carlos III de Madrid\\
\small jrincon@eco.uc3m.es
}
\date{This Version: February 7, 2019}
\begin{document}
\maketitle

\begin{abstract}
In this paper we develop a general framework to analyze stochastic dynamic problems with unbounded utility functions and correlated and unbounded shocks. We obtain new results of the existence and uniqueness of solutions to the Bellman equation through a general fixed point theorem that generalizes known results for Banach contractions and local contractions. We study an endogenous growth model as well as the Lucas asset pricing model in an exchange economy, significantly expanding their range of applicability. 
\end{abstract}  \hspace{10pt}


\section{Introduction} 
Stochastic dynamic programming incorporates uncertain events into a suitable framework to help the decision-maker to design an optimal plan of action. 
A fruitful approach for showing the existence of optimal stationary plans is to prove that the dynamic programming equation admits a unique solution ---the value function--- in a suitable space of functions. See \cite{Blackwell}, \cite{Maitra}, \cite{Furukawa}, \cite{Bertsekas}, \cite {SLP} or \cite{HLLbook}, where this problem is analyzed in detail. There is a huge literature that applies stochastic dynamic programming to economics. \cite{BrockMirman}, \cite{MZ}, \cite{DMehra}, \cite{DD}, \cite{Nyarko}, \cite{HP} or \cite{Mitra} are only a few of the many relevant papers that have contributed to developing this field of research. \cite{OR} makes a review of the contributions to the stochastic optimal growth model.

However, only partial results have been developed so far about the existence and uniqueness of solutions to the dynamic programming equation with arbitrary utility functions and arbitrary shock spaces.
A satisfactory theory of stochastic dynamic programming should include these cases.
\cite{RZRP, RZRP2} developed a method to deal with unbounded utility functions in deterministic problems, the so--called local contraction approach, based on an extension of the Banach contraction principle for function spaces whose topology is defined by a countable family of seminorms.\footnote{\cite{Hadzic} is one of the first papers dealing with this extension. \cite{RZRP}, independently, introduced different hypotheses and applied the results to the deterministic Bellman equation. Recent contributions of the local contraction concept to dynamic programming are \cite{Vailakis}, \cite{MN} and \cite{BRW}.}\ 
The application of the method to the stochastic case is not straightforward if one wants to dispense with artificial bounds on the exogenous shocks. Usually, the stochastic dynamic programming theory imposes those bounds in the form of a compact shock space. This is for analytical convenience only.\footnote{In general, this assumption is incompatible with modelling the Markov chain as a first order stochastic difference equation, even if the underlying i.i.d. shocks takes only finitely many values. Think, for instance, of the simple random walk. It takes every integer number with positive probability. Thus, its range can only be limited by imposing exogenous constraints, alien to the economic model, a limitation which will modify the ``natural" solution and may distort its real implications.}

There are at least two problems when dealing both with unbounded utility functions and with an unbounded shock space. One of them is technical, related to the integrability of the functions involved and thus, to the correctness in the definition of the Bellman operator. A second problem concerns the choice of a suitable family of seminorms (or pseudodistances) that preserves the monotonicity of the Bellman operator. The local contraction approach in the deterministic Bellman equation constructs a family of compact sets in the endogenous state space to define the seminorms, see \cite{RZRP}. When trying to imitate this in the stochastic Bellman equation for the exogenous state space of shocks, one faces the difficulty that the Bellman equation requires the computation of a conditional expectation. This is an averaging of a random variable on the whole exogenous space, breaking down the monotonicity properties of the Bellman operator. 
To overcome the difficulties, we work with an extended concept of contraction parameter(s) in the local contraction definition, to consider an operator that works on the whole family of seminorms. Thus, in this framework, in addition to the selfmap for which we are interested in finding fixed points, there is a companion operator that, acting on the seminorms, plays the role of the contraction parameter of the former selfmap. See Definition \ref{LCOP} below. We state a fixed point theorem, Theorem \ref{Theorem_main}, that applies to this more general framework and we show how it covers previous fixed point results, including the classical Banach ContractionTheorem --- and henceforth, the weighted norm approach in \cite{BoydIII}, or \cite{BB} 
--- as well as those based on local contractions, see e.g., \cite{RZRP} and \cite{Vailakis}.
This idea is not new. \cite{KozlovThimTuresson} developed a fixed point theorem in locally convex spaces whose topology is given by a family of seminorms. However, the results obtained depend on the companion contraction parameter operator being linear. The proof in \cite{KozlovThimTuresson}, of the existence and uniqueness of fixed points, heavily exploits the linearity of the companion operator. This precludes the application of the results of this paper to the dynamic programming equation, as it genuinely demands a nonlinear companion contraction parameter operator, due to the presence of a maximization operation\footnote{It is for this reason that we have to develop our own fixed point theorem, departing from the approach of \cite{KozlovThimTuresson}, as we cannot make use of an equality as its formula (4). Instead, we provide an alternative condition, summarized in our assumption (VI) below.} On the other hand, we simplify somewhat the hypotheses made on the companion operator and extend the result to the consideration of arbitrary pseudodistances ---and hence the topological space is not a locally convex space anymore---, which could be useful for analyzing the unbounded from below case, and more importantly, we show how the dynamic programming equation fits well into this framework, attaining, to our knowledge, new existence and uniqueness results.

The present paper devotes a good deal of efforts to isolating a suitable space of functions and a suitable family of seminorms where the approach explained above is successful.
We find a suitable framework, where the averaging needed to compute the conditional expectation does not break down the monotonicity of the Bellman operator. The seminorms that we define combine the usual supremum norm in the endogenous variables, with an $L^1$ norm in the exogenous variables, and construct a complete space of functions for dealing with the Bellman operator --- a Carath\'eodory function space ---.



The paper is organized as follows. Section \ref{Sect:COP} develops a fixed point theorem for operators acting on topological spaces whose topology is given by a family of pseudodistances that makes it Hausdorff and sequentially complete. The operators enjoy a contraction property, materialized in an associated operator acting on the family of pseudodistances, which plays the role of the contraction parameter in Banach's Contraction Theorem. Thus, the result generalizes Banach's Contraction Theorem and shows that the local contraction approach used in previous papers is a particular case of this more general framework. Section \ref{Sect:BO} applies the theorem to the stochastic dynamic programming equation for models with shocks driven by an exogenous Markov chain and with an unbounded shock space. We carefully choose the set of functions where the Bellman operator is defined and we provide a way to methodically construct the objects needed to apply the fixed point theorem developed in Section \ref{Sect:COP}. We also show that the solution of the Bellman equation is the value function.
 In Section \ref{Sect:Appl}, we study a model of endogenous growth ---which encompasses the one sector optimal growth model---, 
and the Lucas asset pricing model in an exchange economy, in all cases allowing for correlated and unbounded shocks.\ 
Section \ref{Sect:Conclusions} establishes the conclusions of the paper, and give some tips for further research. Appendix \ref{App:Proofs} contains the proofs that are not in the main text, with the exception of the proofs regarding the completeness of the function space and seminorms we consider to analyze the dynamic programming equation, which are developed in Appendix \ref{App:FS}.
Appendix \ref{App:Con} discusses the issue of continuity of the Markov operator appearing in the dynamic programming equation, giving sufficient conditions to establish continuity, and providing an example of non--continuity.


\section{Local Contractions}\label{Sect:COP}

Let $(E,\D)$ be a topological space, where $E$ is a set whose topology is generated by a saturated family of pseudometrics $\D=\{d_a\}_{a\in A}$, with $A$ an arbitrary index set. Since the family $\D$ is saturated, the topology it generates is Hausdorff\footnote{A pseudometric $d:E\times E\rightarrow\R_+$ is a function satisfying $d(x,y)\ge 0$, $d(x,x)=0$, $d(x,y)=d(y,x)$ and $d(x,z)\le d(x,y)+d(y,z)$ for any $x,y,z\in E$, but $d(x,y)=0$ does not imply $x=y$. The family $\D$ of pseudometrics is saturated if $d_a(x,y)=0$ for all $a\in A$, implies $x=y$. Sometimes, the pseudometrics are defined through seminorms $p_a$, $a\in A$, by $d_a(x,y) = p_a(x-y)$, where now $E$ is a real vector space. A seminorm is a function $p:E\rightarrow \R_+$ that satisfies all the axioms to be a norm, except that $p(x)=0$ does not imply that $x$ is the null vector of $E$. If the family of seminorms is saturated, then the topology defined by the family is Hausdorff and the space and $E$ is a locally convex space. See \cite{Willard} for further details.}. We suppose that $(E,\D)$ is sequentially complete: if $\{x_n\}$ is a sequence in $E$ which is Cauchy with respect to all $d_a\in \D$, that is, if $d_a(x_n,x_m)\to 0$ as $n,m\to \infty$, then there is $x\in E$ such that $d_a(x_n,x)\to 0$ as $n\to \infty$ for all $a\in A$.

Given a sequentially complete subset $F\subseteq E$, we study the existence and uniqueness of a fixed point of a mapping $T:F\rightarrow E$.

Let $\R^A$ be the set of functions $d:A\rightarrow \R_+$ and let $\R_+^A$ be the non--negative cone of $\R^A$. On this set we consider the order it generates, that is, for two elements $d,d'\in \R_+^A$, we say that $d\le d'$ if and only if $d(a)\le d'(a)$ for all $a\in A$. The family $\D$ can be embedded into $\R_+^A$, since that, for $x,y\in E$ given, the mapping $a\mapsto d_a(x,y)$ defines a function in $\R_+^A$, that we denote $d^{x,y}(a):=d_a(x,y)$. In general, for a given subset $F\subseteq E$, we let $D(F)$ be the set of functions in $\R^A_+$ which are generated by pairs $x,y\in F$, that is
\[
D(F):=\{d:A\rightarrow \R_+\,:\, d=d^{x,y}\mbox{ for some $x,y\in F$}\}.
\]

\begin{definition}\label{LCOP}
Let $F\subseteq E$. The mapping $T:F\rightarrow E$ is an $L$-local contraction on $F$ with contraction operator parameter $L$ (COP, for short), if there are a set $C\subseteq \R^A_+$ such that $D(F)\subseteq C$, and an operator $L:C\rightarrow \R_+^A$, such that 
\[
d_a(T x, T y) \le (L d^{x,y}) (a),
\]
for all $x,y\in F$ and for all $a\in A$.
\end{definition}

Note that the inequality above can be rewritten $d^{Tx,Ty}\le L d^{x,y}$, that is, as an order relation in the space $\R^A_+$.  The definition of $L$--contractions for mappings $T:F\longrightarrow E$, not imposing $T:F\longrightarrow F$, will facilitate the definition of the COP parameter $L$ of the Bellman operator in Section \ref{Sect:BO}. Of course, the property $T:F\longrightarrow F$ is fundamental for Theorem \ref{GeneralDynPro} below, and will checked carefully in Section \ref{Sect:BO}.

The following two examples show that the operator $L$ is a generalization of the concept of contraction parameter of a (local) contraction mapping.
\begin{example}[Banach contractions]
In the classical Banach's Theorem, $E$ is dotted with a complete metric $d$, so the index set $A$ is a singleton, $\mathcal{D}=\{d\}$, and $T$ is a contraction of constant parameter $\beta$, with $0<\beta<1$: $d(Tx,Ty) \le \beta d(x,y)$, for any $x,y\in E$. The COP is $L=\beta I$, where $I$ is the identity map in $\R_+$.

A generalization of the Banach contraction concept is provided in \cite{Wong}, where it is considered  $T:E\longrightarrow E$ for which there is a function $L :\R_+ \longrightarrow \R_+$ satisfying
\begin{equation}\label{Wong}
d(Tx,Ty) \le L  (d(x,y)),
\end{equation}
for all $x,y\in E$. Note that our definition is an extension of this concept to topological spaces whose topology is given by a family of semidistances. 
\end{example}
\begin{example}[$k$--local contractions]
Suppose that $A=\mathbb{N}$ is countable. In \cite{RZRP, RZRP1}, we introduced the concept of $k$--local contraction in the study of the deterministic Bellman and Koopmans equations, respectively. A $k$--local contraction on $F$, $k=0,1,2,\ldots$, is a mapping $T:F\subseteq E \longrightarrow E$ satisfying
\[
d_j(T x,T y)\le \beta_j d_{j+k}(x,y)
\]
for some fixed sequence of numbers $\{\beta_j\}_{j\in \N}$ with $0<\beta_j<1$, and for all $x,y\in F$.
If we let $s=\R^{\N}$ be the set of real sequences and $s^+$ be the subset of $s$ of nonnegative sequences, then the COP associated with $T$ is the linear operator $L:s^+ \longrightarrow s^+$ acting on sequences given by
\[
L(d_1,d_2,\ldots,d_j,\ldots)=(\beta_1 d_{1+k},\beta_2 d_{2+k},\ldots,\beta_j d_{j+k},\ldots),
\]
where $k\ge 0$ is fixed.
%

Suppose that $A=\mathbb{N}$ is uncountable and let a mapping $\alpha: A\longrightarrow E$ such that for any $a\in A$, $d_a\le d_{\alpha(a)}$. \cite{Vailakis} worked with the following generalization of the countable class above: $T:E\longrightarrow E$ is an $\alpha$-local contraction if there exists a function $\beta: A\longrightarrow [0,1)$ such that
\[
d_a(T x,T y)\le \beta (a) d_{\alpha(a)} (x,y).
\]
The COP $L$ acts as follows: given a function $d:A\longrightarrow\R_+$, the image function is $(L d)(a) = \beta(a) d(\alpha(a))$, that is, a translation in the independent variable by $\alpha$, and a multiplication by $\beta$. It turns out that $L$ is also a linear mapping, as in the countable case above.
\end{example}

In what follows, we use the standard notation for successive iterations of the operators $T$ and $L$. For instance, $L^0$ is the identity operator on $C$, $L^1=L$, and for $t\ge 2$, $L^t=L\circ L^{t-1}$.
We impose to $C$, $L$ and $T$ the assumptions (I) to (VI) listed below. The assumptions (I) to (V) concern the behavior of $L$ on the set $C$. Assumption (VI) links directly the operators $T$ and $L$.
\strut

\begin{enumerate}[{(}I{)}]
\item \it $D(F)\subseteq C$. For all $d,d'\in C$, the sum $d+d'\in C$, and any bounded subset of $C$  is countably chain complete\footnote{A subset $S\subseteq C$ is bounded with respect to the order inherited from $\R^A$ if there is $d'\in C$ such that $d\le d'$ for all $d\in S$. The bounded subset $S$ is countably chain complete if for any countably chain $d_1\le d_2\le \cdots d_t\le\cdots$ in $S$, $\sup_{t\in \N} d_t\in S$.} Moreover, if $d'\in C$, $d\in \R_+^A$ and $d\le d'$, then $d\in C$.\rm

\item \it $L0= 0$.\rm

\item \it $L$ is monotone: for all $d,d'\in C$ with $d\le d'$, $L d \le L d'$.\rm

\item \it $L$ is subadditive: for any $d, d'\in C$\rm
\[
L (d+d') \le L d + L d'.
\]
\item \it $L$ is upper semicontinuous sup-preserving\footnote{For instance, the sup-preserving property, $L (\sup_t d_t) = \sup_t L  d_t$, plays a prominent role in the Fixed Point Theorem of Kantorovich-Tarski. In our context, it can be weakened to a kind of upper semicontinuity.}: for any bounded countably chain in $C$,
$d_1\le d_2\le \cdots\le d_t\le \cdots$,
\[
L \sup_t d_t \le  \sup_t L  d_t.
\]\rm
\end{enumerate}



\begin{enumerate}[{(}VI{)}]
\item \it There are $x_0\in F$ and $r:0\in C$ with $d_a(x_0,Tx_0)\le r_0(a)$ and
\[
R_0(a):=\sum_{t=0}^{\infty} L^t r_0(a)<\infty,
\]
for all $a\in A$.
\end{enumerate}

\strut

Since $L^tr_0\in C$, for all $t=0,1,\ldots$, and the countable chain $\{r_0, r_0+Lr_0,\dots,r_0+Lr_0+\cdots +L^t r_0,\cdots\}$ is bounded in $C$ by (VI), $R_0$ is in $C$ by assumption (I).


For $F\subseteq E$, $x_0\in F$, and $m\in \R_+^A$, let the set
\[
V_F(x_0,m)=\{x\in F\,:\, d_a(x_0,x) \le m(a), \text{ $\forall a\in A$}\}.
\]
When $E$ is a metric space, that is, when $A$ is a singleton, the pseudometric is a metric, and $V_F(x_0,m)$ is simply the intersection with $F$ of the closed ball centered at $x_0$ and radius $m$.

\begin{lemma}\label{B_defined}
Let $T:F\longrightarrow F$ be an $L$-local contraction on $F\subseteq E$ and let $x_0\in F$ be such that (I)--(VI) hold true for a suitable $r_0\in C$. Let $R_0$ be defined as in (VI). Then
\begin{enumerate}[{(a)}]
\item $T: V_F(x_0,R_0) \longrightarrow V_F(x_0,R_0)$.
\item For any $a\in A$, $\lim_{t\to \infty} (L^t R_0) (a) = 0$.
\end{enumerate}
\end{lemma}
The following result is a fixed point theorem for $L$-local contractions.
\begin{theorem}\label{Theorem_main}
Let $(E,\D)$ be a Hausdorff and sequentially complete topological space. Let $T:F\rightarrow F$ be an $L$-local contraction on the sequentially complete subset $F\subseteq E$ and let $x_0\in F$ be such that (I)--(VI) hold true. Then there is a unique fixed point $x^*\in V_F(x_0, R_0)$ of $T$, which is the limit of any iterating sequence $y_{t+1} = T y_t$, $t=0,1,2,\ldots$, where $y_0=x\in V_F(x_0,R_0)$ is arbitrary.
\end{theorem}
\begin{proof}
Consider first the iterating sequence $x_{t+1} = T x_t$, $t=0,1,2,\ldots$ (that is, the initial seed is $x_0$). By Lemma \ref{B_defined}, $x_t$ is in $V_F(x_0,R_0)$ for any $t=0,1,2,\ldots$. Since $T$ is an $L$-local contraction
\[
d_a(x_{t},x_{t+1}) = d_a(T x_{t-1}, Tx_{t}) \le L d^{x_{t-1},x_{t}} (a)
\]
and by induction $d_a(x_{t},x_{t+1}) \le ( L^t d^{x_0,Tx_0} )(a)$.
Let $r>s\ge 1$. Then by the triangle inequality extended to finite sums
\begin{equation}\label{estimate}
d_a(x_{s},x_{r+1}) \le \sum_{t=s}^{r} d_a (x_t,x_{t+1}) \le  \sum_{t=s}^{r} L^t d^{x_0,Tx_0} (a)\le R_0(a)<\infty,
\end{equation}
for all $a\in A$. By the Cauchy criterion for series, $d_a(x_{s},x_{r+1})$ tends to 0 as $r,s\to \infty$. Since $a$ is arbitrary, the sequence $\{x_t\}_{t=0}^{\infty}$ is Cauchy, hence it converges to some $x^*\in F$, as $F$ is sequentially complete. In fact, $x^*\in V_F(x_0,R_0)$. To see this, note that for all $a\in A$, the mapping $x\mapsto d_a(x_0,x)$ is trivially sequentially continuous in the topology generated by $\D$, hence $d_a(x_0,x^*)=d_a(x_0,\lim_{t\to \infty} x_t)=\lim_{t\to\infty} d_a(x_0,x_t)\le R_0 (a)$ for any $a\in A$. Next, we prove that $x^*$ is a fixed point of $T$. By estimate \eqref{estimate}, letting $r\to \infty$, $d_a(x_{s},x^*) \le \sum_{t=s}^{\infty} L^t d^{x_0,Tx_0} (a) = \sup_N \sum_{t=s}^{N} L^t d^{x_0,Tx_0} (a) $. Hence
\begin{align*}
L d^{x_s,x^*}(a) &\le  L\Big(\sup_N \sum_{t=s}^{N} L^t d^{x_0,Tx_0}\Big) (a) \\
&\le \sup_N L \Big(\sum_{t=s}^{N} L^t d^{x_0,Tx_0}\Big) (a) \\
&\le \sup_N \sum_{t=s}^{N} L^{t+1} d^{x_0,Tx_0} (a),
\end{align*}
which tends to 0 as $s\to \infty$, since $R_0$ is finite. Hence $L d^{x_s,x^*}\to 0$ as $s\to \infty$, for all $a\in A$. The first line above is due to the monotonicity of $L$, the second line since $L$ is upper sup-preserving, and the third one, since $L$ is subadditive. Now, given that
\[
d_a(x^*,Tx^*) \le d_a(x_{s+1},x^*) + d_a(x_{s+1},Tx^*) \le d_a(x_{s+1},x^*) + Ld^{x_s,x^*}(a)
\]
and that both summands tend to 0 as $s\to \infty$, we conclude that $x^*=Tx^*$.
To prove uniqueness, we argue by contradiction, supposing the existence of another fixed point $x^{**}\in V_F(x_0,R_0)$. Then $d_a(x^*, x^{**})\le 2 R_0(a)$ for all $a\in A$, and hence
\[
d_a(x^*, x^{**}) = d_a(T^tx^*, T^tx^{**}) \le L^t d^{x^*, x^{**}} (a) \le 2 L^t R_0 (a),
\]
since $L$ is both monotone and subadditive\footnote{It is easy to prove that $L^t$ is monotone and subadditive for any $t$.}. Taking the limit as $t\to \infty$ and using Lemma \ref{B_defined}, $x^* = x^{**}$ is proven. Finally, let $x\in V_F(x_0,R_0)$ and let the iterating sequence $y_{t+1} = T y_t$, for $n=0,1,2,\ldots,$, with $y_0=x$. Observe that
\[
d_a(x_t,y_t) \le L^t d^{x_0,x}(a) \le 2 L^t R_0(a) \to 0 \text{ as $n\to \infty$,}
\]
by Lemma \ref{B_defined}, so $d_a(y_{t}, x^*)\le d_a(x_t,y_t) + d_a(x_t,x^*)$ tends to 0 as $t\to \infty$.
\end{proof}

The next corollary provides conditions for the uniqueness of the fixed point in $F$ and not only in $V_F(x_0,R_0)$. 
When $T$ is indeed an $L$-local contraction on the whole $E$, this result provides global uniqueness of the fixed point on $E$.
\begin{theorem}\label{GlobalU}
Let $(E,\D)$ be a Hausdorff, sequentially complete space. Let $T:F\rightarrow F$ be an $L$-local contraction on the sequentially complete subset $F\subseteq E$ and let $x_0\in F$ be such that (I)--(VI) hold true. Suppose that, for any $x\in F$, it is possible to choose $r_0\in C$ satisfying (VI), such that $x\in V_F(x_0,R_0)$. Then there is a unique fixed point of $T$ in $F$ and convergence to the fixed point of successive iterations of $T$ is attained from any $x\in F$.
\end{theorem}
\begin{proof}
By Theorem \ref{Theorem_main}, $T$ admits a unique fixed point $x^*$ in $V_F(x_0,R_0)$, where $R_0=\sum_{t=0}^{\infty}L^t r_0$, for any $r_0\in C$ for which $R_0$ is a convergent series. Suppose, by contradiction, that $T$ admits another fixed point $x^{**}\neq x^*$ in $F$. By assumption, there is $r'_0\in C$ such that $x^{**}\in V_F(x_0,R'_0)$, where $R_0' = \sum_{t=0}^{\infty} L^tr_0'$ is finite. Hence, $x^*=x^{**}$. The convergence of iterating sequences is also an immediate consequence of Theorem \ref{Theorem_main}.
\end{proof}

Next we establish a useful sufficient condition for (VI). Note that the Bellman operator satisfies the extra condition imposed on $L$.
\begin{proposition}\label{SufficientB6}
Let $(E,\D)$ be a Hausdorff and sequentially complete topological space. Let $T:F\longrightarrow F$ be an $L$--local contraction on $F\subseteq E$, with COP $L$ satisfying (I) to (V) and
$
L(\alpha d) \le \alpha L d$, for all $d\in C$, for all $ \alpha \in[0,1]$.
Let $x_0\in F$, for which there is $t_0\in \{0,1,2,\ldots\}$, $s\in C$, and $\theta\in[0,1)$ such that
\begin{equation}\label{Lalphad}
L^{t_0} d_0 \le s\quad \mbox{and}\quad Ls\le \theta s,
\end{equation}
where $d_0(a)=d_a(x_0,Tx_0)$. Then (VI) holds with $r_0 = d_0$.
\end{proposition}
\begin{proof}
Note that $\sum_{t=t_0}L^td_0\le s + L s+ L^2s+\cdots\le (1+\theta+\theta^2+\cdots)s=\frac 1{1-\theta}s$. Hence, $\sum_{t=0}^{\infty} L^td_0 =\sum_{t=0}^{t_0-1}L^td_0 + \sum_{t=t_0}^{\infty} L^td_0\le \sum_{t=0}^{t_0-1}L^td_0 +\frac 1{1-\theta}s$ is finite for all $a\in A$.
\end{proof}


\section{Stochastic Dynamic Programming and Bellman Equation}\label{Sect:BO}
Consider a dynamic programming model $(X,Z,\Gamma,Q,U,\beta)$, where
$X\times Z$ is the set of possible states of the system,
$\Gamma$ is a correspondence that assigns a nonempty set $\Gamma(x,z)$ of feasible actions to each state $(x,z)$ and
 $Q$ is the transition function, which associates a conditional probability distribution $Q(z,\cdot)$ on $Z$ to each $z\in Z$. Hence, the law of motion is assumed to be a first-order Markov process, which could be degenerated, giving rise to a deterministic model. We will use indistinctly the notation $Q_z(\cdot)=Q(z,\cdot)$; the function
$U$ is the one--period return function, defined on the graph of $\Gamma$, $\Omega=\{(x,y,z)\,:\, (x,y)\in X\times Z, y\in \Gamma(x,z)\}$, and
$\beta$ is a discount factor.

Starting at some state $(x_0,z_0)$, the agent chooses an action $x_1\in \Gamma(x_0,z_0)$, obtaining a return of $U(x_0,x_1,z_0)$ and the system moves to the next state $(x_1,z_1)$, which is drawn according to the probability distribution $Q(\cdot|z_0)$. Iteration of this process yields a random sequence $(x_0,z_0,x_1,z_1,\ldots)$ and a
total discounted return $\sum_{t=0}^{\infty} \beta^t U(x_t, x_{t+1},z_t)$. A history of length $t$ is $z^t=(z_0,z_1,\ldots, z_t)$. Let $Z^t$ be the set of all histories of length $t$. A (feasible) plan $\pi$ is a constant value $\pi_0\in X$ and a sequence of measurable functions $\pi_t : Z^t \longrightarrow X$, such that $\pi_t(z^t)\in \Gamma(\pi_{t-1}(z^{t-1}),z_t)$, for all $t=1,2,\ldots$. Denote by $\Pi(x_0,z_0)$ the set of all feasible plans starting at the state $(x_0,z_0)$. Any feasible plan $\pi\in \Pi(x_0,z_0)$, along with the transition function $Q$, defines a distribution $\mathbb{P}^{\pi,(x_0,z_0)}$ on all possible futures of the system $\{(x_t,z_t)\}_{t=1}^{\infty}$, as well as the expected total discounted utility
\[
u(\pi,x_0,z_0) = \mathbb{E}^{\pi,(x_0,z_0)}\left(\sum_{t=0}^{\infty} \beta^t U(x_t,x_{t+1},z_t)\right).
\]
The expectation $\mathbb{E}^{\pi,(x_0,z_0)}$ is taken with respect to the distribution $\mathbb{P}^{\pi,(x_0,z_0)}$.
The problem is then to find a plan $\pi\in \Pi(x_0,z_0)$ such that $u(\pi,(x_0,z_0))\ge u(\widehat \pi,(x_0,z_0))$ for all $\widehat \pi \in \Pi(x_0,z_0)$, for all $(x_0,z_0)\in X\times Z$. The value function of the problem is
$v(x_0,z_0)=\sup_{\pi\in\Pi(x_0,z_0)} u(\pi, (x_0,z_0))$.

Consider the functional equation corresponding to the above dynamic programming problem as stated in \cite{SLP}.
For $x\in X$, $z\in Z$
\begin{equation}\label{Bellmaneq}
v(x,z)=\max _{y\in \Gamma(x,z)} \left\{U(x,y,z) + \beta \int_Z v(y,z')Q(z,dz')\right\}.
\end{equation}

A solution of the Bellman equation satisfying additional assumptions is the value function of the infinite programming problem. This is the content of Theorem \ref{GeneralDynPro} below, whose proof needs
the notion of the probability measure $\mu^t$ defined on the sequence space of shocks $(Z^t, \mathcal{Z}^t)$ for finite $t=1,2,\ldots$, where
\[
(Z^t, \mathcal{Z}^t) = (Z\times \cdots\times Z, \mathcal{Z}\times \cdots \times\mathcal{Z})\quad (\mbox{$t$ times}).
\]
For any rectangle $B=A_1\times \cdots\times A_t\in \mathcal{Z}^t$, $\mu^t$ is defined by
\[
\mu^t(z_0,B) = \int_{A_1}\ldots\int_{A_{t-1}}\int_{A_t} Q_{z_{t-1}}(dz_t) Q_{z_{t-2}}(dz_{t-1}) \cdots Q_{z_{0}}(dz_1),
\]
and by the Hahn Extension Theorems, $\mu^t(z_0,\cdot)$ has a unique extension to a probability measure on all of $\mathcal{Z}^t$. We omit the details, that can be found in \cite{SLP}, Section 8.2, whose presentation we follow closely.

Defining the Bellman operator in a suitable function space $E$, such that for $f\in E$
\[
(Tf)(x,z) = \max _{y\in \Gamma(x,z)} \left\{U(x,y,z) + \beta \int_Z f(y,z')Q(z,dz')\right\},
\]
the Bellman functional equation \eqref{Bellmaneq} is a fixed point problem for $T$. This fixed point problem is completely understood for the case where $U$ is bounded. There are now also different approaches for some special cases for unbounded $U$. It is worth mentioning the constant returns to scale model and the logarithmic and the quadratic parametric examples analyzed in \cite{SLP}, pp. 270--280, and the weighted norm approach in \cite{BoydIII} and \cite{HLLbook}. One feature of all these approaches is that they consider a bounded (or compact) space of shocks, an assumption that we want to dispense with.\footnote{The weighted norm approach presents some limitations, which are explained, for instance, in Remark 9 of \cite{MN}. This paper constitutes a first attempt to translate the approach initiated by \cite{RZRP} for deterministic programs to the stochastic case. However, the results obtained do not cover a general model where shocks are driven by an exogenous transition probability. In fact, in the class of dynamic programming models described here and in \cite{SLP}, assumption (A4) on the generation of shocks $\{z_t\}_{t=0}^{\infty}$ imposed in \cite{MN}, basically implies that the space of shocks $Z$ is compact, or that the underlying probability has compact support.} Allow for a non--compact shock space is important for a qualitative analysis of models, see for instance \cite{BinderPesaran} and \cite{Stach}, and more recently, \cite{MaStachurski}.

We now impose the standing hypotheses. Most are taken from \cite{SLP}, but there are essential differences, as we admit an unbounded utility $U$ and an unbounded shock space $Z$.
\begin{enumerate}[{(B}1{)}]
\item $X\subseteq \R^l$,  $Z\subseteq \R^k$ are Borel sets, with Borel $\sigma$-algebra $\mathcal{X}$ and $\mathcal{Z}$, respectively. The set $X$ is endowed with the Euclidean topology.
\item $0<\beta<1$.
\item $Q:Z\times \mathcal{Z}\rightarrow [0,1]$ satisfies
\begin{enumerate}
\item for each $z\in Z$, $Q(z,\cdot)$ is a probability measure on $(Z,\mathcal{Z})$; and
\item for each $B\in \mathcal{Z}$, $Q(\cdot,B)$ is a Borel measurable function.
\end{enumerate}
\item The correspondence $\Gamma : X \times Z \longrightarrow X$ is nonempty, compact-valued and continuous. 
\item $U: \Omega \longrightarrow \R$ is a Carath\'eodory function, that is, it satisfies
\begin{enumerate}
\item for each $(x,y)\in D := \{(x,y)\in X\times Y\,:\, \exists z\in Z, \ y\in \Gamma (x,z)\}$, the function of $z$, $U(x,y,\cdot): Z \longrightarrow \R$ is Borel measurable;
\item for each $z\in Z$, the function of $(x,y)$, $U(\cdot, \cdot,z):D\longrightarrow \R$ is continuous.
\end{enumerate}
\end{enumerate}

The reason for working with Carath\'eodory functions instead of continuous functions in the three variables $(x,y,z)$ is twofold. On the one hand, the Markov operator
\begin{equation}\label{MarkovO}
(Mf) (x,z) := \int _Z f(x,z') Q(z,dz'),
\end{equation}
does not preserve continuity of $f$, if $f$ is continuous but not bounded, as the simple example in Appendix \ref{App:Con} shows.

%
On the other hand, the Bellman operator is well defined for the class of Carath\'eodory functions in the unbounded case, while working with the supremum norm is not possible. A direct attack of the Bellman equation in the space of $(x,z)$--continuous functions does not work for unbounded functions and/or unbounded shock space: known theorems on local contractions---with a countable or uncountable index set ---are not suitable, due to the averaging operation involved in the computation of conditional expectations. For this reason 
we are going to use $L^1$-type seminorms, whose precise definition is given below.

We now describe the function space, which details are given in Appendix \ref{App:FS}.
For each $z\in Z$, let $L^1(Z,\mathcal{Z},Q_z)$ be the space of Borel measurable functions\footnote{It is well known that $L^1(Z,\mathcal{Z},Q_z)$ consists of equivalence classes rather than functions, identifying functions that are equal $Q_z$--almost everywhere.} $g:Z\longrightarrow \R$ such that
$
 \int_Z |g(z')|Q_z(dz')<\infty
$.
In what follows, we let $\mathcal{K}$ be the family of all compact subsets of $X$.

Consider the space $E:= \mathcal{L}^1(Z;C(X))$, formed by Carath\'eodory functions $f:X\times Z\longrightarrow  \R$ such that the function $z'\mapsto \max_{x\in K} |f_x(z')|$ is in $L^1(Z,\mathcal{Z},Q_z)$, for all compact sets $K\in \mathcal{K}$, and all $z\in Z$.
See Appendix \ref{App:FS} for the definitions and the notation, where it is also proved the following fundamental result.

\begin{lemma}\label{lemma:function_space}
$E= \mathcal{L}^1(Z;C(X))$ is a complete locally convex space with the topology generated by the family of seminorms $\mathcal{P}:=\{p_{K,z}\}_{K\in \mathcal{K}, z\in Z}$, given by
\begin{equation}\label{pKz}
p_{K,z} (f) :=\int_Z  \max_{x\in K} |f(x,z')| Q_z(dz').
\end{equation}
\end{lemma}
In particular, the lemma states that $E$ is sequentially complete.
In the notation of Section \ref{Sect:COP}, the index set of the family of seminorms is $A=\mathcal{K}\times Z$.

Given a solution $f\in \mathcal{L}^1(Z;C(X))$ of \eqref{Bellmaneq}, define the policy correspondence $G^f:X\times Z \rightarrow 2^X$ by
\begin{equation}\label{Gamma*}
G^f(x,z)=\{y\in \Gamma(x,z)\,:\, f(x,z)=U(x,y,z) + \beta Mf(y,z)\}.
\end{equation}
This is the optimal policy correspondence, denoted simply by $\Gamma^*$, when $f$ is the value function, $v$.

Remember from Section \ref{Sect:COP}, that for a subset $F \subseteq E$, the set $D(F)$ is in this context
\[
D(F)=\{p:\mathcal{K}\times Z\rightarrow \R_+\,:\, p(K,z)=p_{K,z}(f) \mbox{ for some $f\in F$}\}.
\]
\begin{notation}\label{notation1}Along the paper, we will use the notation
\[
\psi(x,z) = \max_{y \in \Gamma(x,z)} U(x,y,z) = T0(x,z)
\]
while, for $p:\mathcal{K}\times Z \longmapsto \R_+$, the function  $p[\Gamma]: X\times Z \longmapsto \R_+$ is defined by $p[\Gamma](x,z)= p(\Gamma(x,z),z)$, that is, it is the function of $(x,z)$ obtained through $p$, when the compact sets $K$ equal $\Gamma(x,z)$, for $x\in X$, $z\in Z$.
\end{notation}

The next result shows that $T$ is an $L$-local contraction, and gives the expression of $L$: Given $p:\mathcal{K}\times Z \mapsto \R_+$ for which $p[\Gamma]\in \mathcal{L}^1(Z;C(X))$, the operator $L$ computes the seminorm of the function $p[\Gamma]$, that is, $(Lp) (K,z) = \beta p_{K,z} (p[\Gamma])$. Note that $L$ is \emph{nonlinear}. The expanded definition of the operator $L$ is the expression \eqref{eq:LCOP} below.

\begin{proposition}\label{prop:LCOP}
Let the Bellman operator $T:F\longrightarrow E$, where $F\subseteq \mathcal{L}^1(Z;C(X))$, such that for all $p\in D(F)$, $p[\Gamma]\in \mathcal{L}^1(Z;C(X))$. Then, $T$ is 
an $L$-local contraction on $F$ with COP $L:D(F)\longrightarrow \R_+^{\mathcal{K}\times Z}$ given by
\begin{equation}\label{eq:LCOP}
(Lp)(K,z) = \beta \int _Z  \max_{x\in K} p(\Gamma(x,z'),z') Q_z(dz'), 
\end{equation}
for all $K\in \mathcal{K}$ and $z\in Z$.
\end{proposition}
\begin{proof}
Following \cite{Blackwell}, we exploit the fact that $T$ is monotone, in conjunction with the properties of the seminorms $p_{K,z}$. Let $f,g\in E$ and let $x\in X$, $K\in \mathcal{K}$ and $z\in Z$. Let $y\in \Gamma (x,z)$ and $z'\in Z$ arbitrary. Then $f(y,z')\le g(y,z') + |f(y,z')-g(y,z')|$ implies
$
f(y,z')\le g(y,z') + \max_{y\in \Gamma(x,z)} |f(y,z')-g(y,z')|
$
and then, by monotonicity and linearity of the integral,
\begin{align*}
\int_Z f(y,z')  Q_z(dz') \le &\int_Z g(y,z')  Q_z(dz') \\
&\quad + \int_Z \max_{y\in \Gamma(x,z)} |f(y,z')-g(y,z')| Q_z(dz').
\end{align*}
We are allowed to take the integral by Lemma \ref{welldefined}.
The inequality is maintained after multiplying by $\beta$ and adding $U(x,y,z)$ to both sides. Then, by taking the maximum in $y\in \Gamma (x,z)$ to both sides, we have
\begin{align*}
(Tf)(x,z) &\le (Tg)(x,z) + \beta \max _{y\in \Gamma(x,z)} \int_Z \max_{y\in \Gamma(x,z)} |f(y,z')-g(y,z')| Q_z(dz')\\
& =  (Tg)(x,z) + \beta \int_Z \max_{y\in \Gamma(x,z)} |f(y,z')-g(y,z')| Q_z(dz')\\
& =  (Tg)(x,z) + \beta p_{\Gamma(x,z),z}(f-g).
\end{align*}
Exchanging the roles of $f$ and $g$, we have
\[
|(Tf)(x,z)-(Tg)(x,z)|\le \beta  p_{\Gamma(x,z),z}(f-g).
\]
It is convenient to write this inequality with the dummy variable $z'$ instead of $z$. Now, taking the maximum in $x\in K$ and averaging with respect to the measure $Q_z$, we obtain
\[
\int_Z \max_{x\in K} |(Tf)(x,z')-(Tg)(x,z')|Q_z(dz') \le \beta \int_Z \max_{x\in K} p^{f-g}(\Gamma(x,z'),z')  Q_z(dz').
\]
This inequality can be rewritten $
p_{K,z}(Tf-Tg) \le (L p^{f-g}) (K,z)
$,
for all $K\in \mathcal{K}$, $z\in Z$, where $L$ is the operator defined in \eqref{eq:LCOP}.
\end{proof}

One of the difficulties in applying contraction techniques to the dynamic programming equation, when the return function and/or the space of shocks is unbounded, is the selection of a suitable space of functions where the Bellman operator is a selfmap. Assumption (B6) below provides a scheme to construct such a space along the lines of assumption (VI) in Section \ref{Sect:COP}.
This is in the same spirit of Assumption 9.3 in \cite{SLP}, pp. 248-249.\ 
Our assumption is not about bounding the one-shot utility function $U$ along any policy path by a function that depends only on time and the initial state, but about bounding its expected value with respect to the initial state. This is an important difference, as it allows us to deal with an unbounded space of shocks. 

\strut

\begin{enumerate}[{(B}6{)}]
\item
There is a collection of nonnegative functions $\{l_t\}_{t=0}^{\infty}\in \mathcal{L}^1(Z;C(X))$, such that for all $x\in X$, for all $z\in Z$
\[
\begin{array}{l}
 l_0(x,z) \ge |\psi(x,z)|;\\[1ex]
l_{t+1}(x,z)\ge 
\beta {\displaystyle \int_Z} \max_{y\in \Gamma(x,z)} l_t(y,z')Q_z(dz'),\quad \mbox{for all $t=0,1,\ldots$,}
\end{array}
\]
and the series $w:= \sum_{t=0}^{\infty}  l_t$
is unconditionally convergent, that is,
\[
R(K,z):=\sum_{t=0}^{\infty}  p_{K,z} (l_t) <\infty,
\]
for all $K\in \mathcal{K}$, for all $z\in Z$.
\end{enumerate}

\strut


Now we consider a suitable set $C$ where $L$ is defined.
\begin{equation}\label{SetC}
\begin{aligned}
C=\Big\{p:\mathcal{K}\times Z \longmapsto \R_+\,: \,p(K,z) \le c R_0(K,z) \mbox{ for some $c>0$,}&\\
\mbox{ and } p[\Gamma]\in \mathcal{L}^1(Z,C(X))\Big\}.&
\end{aligned}
\end{equation}
As it is proved in Lemma \ref{lemma:C}, $C$ is not trivial, as it contains the images of $V(0,R_0)$ by the family of seminorms $\mathcal{P}$.


Theorem \ref{GeneralDynPro} below is a fixed point theorem for the Bellman operator with unbounded utility and unbounded space of shocks. We state a previous lemma.
\begin{lemma}\label{Preparation_Th}
Let assumptions (B1) to (B6) hold. Then $T$ and $L$ with $C$ defined in \eqref{SetC}, satisfy (I) to (VI).
\end{lemma}

\begin{theorem}\label{GeneralDynPro}
Let assumptions (B1) to (B6) hold. The following is true.
\begin{enumerate}[{(a)}]
\item The Bellman equation admits a unique solution $v^*$ in $V(0,R_0)$.
\item If the correspondence $G^{v*}$ defined in \eqref{Gamma*} admits a measurable selection, then $v^*$ coincides with the value function, $v=v^*$, and for all $v_0\in V(0,R_0)$, $T^nv_0\to v$ as $n\to \infty$, that is, $p_{K,z}(T^nv_0-v) \to 0$, for all $K\in \mathcal{K}$ and $z\in Z$. Moreover, for all $z\in Z$, the optimal policy correspondence $\Gamma^*(\cdot,z): X\rightarrow X$ is non-empty, compact valued and upper hemicontinuous.
\end{enumerate}
\end{theorem}
\begin{proof} (a) $T$ is an $L$--contraction by Proposition \ref{prop:LCOP} and all the assumptions of Theorem \ref{Theorem_main} hold true by Lemma \ref{Preparation_Th}. Hence $T$ admits a unique fixed point $v^*$ is $V(0,R_0)$ and the rest of conclusions of Theorem \ref{GeneralDynPro} hold true.

(b)
To see that $v^*$ is the value function of the problem, we invoke Theorem 9.2 in \cite{SLP}.
Recall that, for any function $F$ that is $\mu^t(z_0,\cdot)$-integrable, its conditional expectation can be expressed as
\begin{align*}
\E_{z_0}(F) &:= \int_{Z^t} F(z^t) \mu^t(z_0,dz^t)\\
& = \int_{Z^{t-1}}\left[ \int_Z F(z^{t-1},z_t) Q_{z_{t-1}}(dz_t)\right]\mu^{t-1}(z_0,dz^{t-1})\\
& =  \int_Z\left[ \int_{Z^{t-1}}F(z_1,z_2^t) \mu^{t-1}(z_1,dz_2^t)\right]Q_{z_0}(dz_1).
\end{align*}
The assumptions of Theorem 9.2 in \cite{SLP} are: (i) $\Gamma$ is non-empty valued, with a measurable graph and admits a measurable selection; (ii) for each $(x_0,z_0)$ and each feasible plan $\pi$ from $(x_0,z_0)$, $U(\pi_{t-1}(z^{t-1}),\pi_t(z^t),z_t)$ is $\mu^t(z_0,\cdot)$-integrable, $t=1,2,\ldots$, and the limit
\begin{equation}\label{Limii}
U(x_0,\pi_0,z_0) + \lim_{n\to \infty} \sum_{t=1}^{n} \int_{Z^t} \beta ^t U(\pi_{t-1}(z^{t-1}),\pi_t(z^t),z_t)\mu^t(z_0,dz^t)
\end{equation}
exists; and (iii) $\lim_{t\to \infty}  \int_{Z^t} \beta ^t v^*(\pi_{t-1}(z^{t-1}), z_t)\mu^t(z_0,dz^t) = 0$.

(i) is implied by (B5) and (ii) is implied by (B6), since $|U(\pi_{t-1}(z^{t-1}),\pi_t(z^t),z_t)|$ is clearly measurable, given that $U$ is a Carath\'eodory function. Moreover, since $\l_0$ in (B6) is in $\ca{X\times Z}$, we can apply Fubini's Theorem so that $l_0(\pi_{1}(z^1),z_2)$ is $\mu^2(z_0,\cdot)$-integrable and
\begin{align*}
\int_{Z^2} l_0(\pi_{1}(z^1),z_2)\mu^2(z_0,dz^2) &= \int_{Z^1}\left(\int_Zl_0 (\pi_{1}(z^1),z_2)Q_{z_1}(dz_2)\right)\mu^1(z_0,dz^1)\\
&\le \int_{Z^1}\frac 1{\beta} l_1 (\pi_{0}(z_0),z_1)\mu^1(z_0,dz^1)\\
&\le \frac 1{\beta^2} l_2(x_0,z_0).
\end{align*}
Both inequalities are due to assumption (B6). By induction, we get that $l(\pi_{t-1}(z^{t-1}),z_t)$ is $\mu^t(z_0,\cdot)$--integrable and
\[
\int_{Z^t} l_0(\pi_{t-1}(z^{t-1}),z_t)\mu^t(z_0,dz^t) \le \frac 1{\beta^t} l_t(x_0,z_0).
\]
Since
$ |U(\pi_{t-1}(z^{t-1}),\pi_t(z^t),z_t)|\le  l_0(\pi_{t-1}(z^{t-1}),z_t)$, the first part of (ii) is proved. Indeed, this estimate provides the bound
\begin{align*}
&|U(x_0,\pi_0(z_0),z_0)| + \sum_{t=1}^{n}  \int_{Z^t} \beta ^t |U(\pi_{t-1}(z^{t-1}),\pi_t(z^t),z_t)|\mu^t(z_0,dz^t)\\
&\quad  \le |U(x_0,\pi_0(z_0),z_0) | + \sum_{t=1}^{n} l_t(x_0,z_0)\le w_0(x_0,z_0),
\end{align*}
hence the second part of (ii) also holds, that is, the limit \eqref{Limii} is finite. Moreover, since the above inequality holds for any $\pi\in \Pi(x_0,z_0)$, it shows that the $n$--th iteration of $T$ on the null function as the initial seed satisfies $|T^n0(x_0,z_0)|\le w_0(x_0,z_0)$. Hence, since $\int_Z |T^n0(x_0,z_1) - v^*(x_0,z_1)|Q_{z_0}(dz_1)$ tends to 0 as $n\to \infty$, by part (a) above, we obtain the bound
\begin{equation}\label{v*}
\int_Z |v^*(x_0,z_1)|Q_{z_0}(dz_1) \le \int _Z w_0(x_0,z_1)Q_{z_0}(dz_1).
\end{equation}
This inequality will be used to show (iii). First, we claim that for any $t$, for any $\pi\in \Pi(x_0,z_0)$,
\[
\int_{Z^t} \beta^t w_0(\pi_{t-1}(z^{t-1}),z_t) \mu^t(z_0,dz^t)\le \sum_{s=t}^{\infty} l_s(x_0,z_0).
\]
To prove it, we employ mathematical induction. Let $t=1$. Then, by assumption (B6)
\begin{align*}
\int_{Z} \beta w_0(\pi_{0}(z_0),z_1) \mu^1(z_0,dz^1) &=\int_Z  \beta  \sum_{t=0}^{\infty}  l_t(\pi_{0}(z_0),z_1) Q_{z_0}(dz_1) \\ 
&= \sum_{t=0}^{\infty} \beta   \int_Z l_t(\pi_{0}(z_0),z_1) Q_{z_0}(dz_1)\\
&\le  \sum_{t=0}^{\infty} l_{t+1}(x_0,z_0).
\end{align*}
The exchange of the integral and infinite sum is possible by the Monotone Convergence Theorem. Suppose that the property is true for $t$ and let us prove it for $t+1$. Then it will hold for any $t$. Note
\begin{align*}
\int_{Z^{t+1}} \beta^{t+1} &w_0(\pi_{t}(z^{t}),z_{t+1}) \mu^{t+1}(z_0,dz^{t+1})\\
 &=
\int_Z \left(\beta \int_{Z^t} \beta^t w_0(\pi_{t-1}(z^{t-1}),z_t) \mu^t(z_0,dz^t)\right) Q_{z_0}(dz_1)
\\
&\le \int_Z \beta \sum_{s=t}^{\infty} l_s(\pi_{0}(z_0),z_1) Q_{z_0}(dz_1)\\
&\le \sum_{s=t+1}^{\infty} l_s(x_0,z_0),
\end{align*}
again by the Monotone Convergence Theorem, and where we have used Fubini's Theorem and the induction hypothesis. This and \eqref{v*} imply (iii), since the series $w_0$ converges. Thus, $v^*$ is the value function.
The claims about $\Gamma^*$ are 
immediate from the Theorem of the Maximum of Berg\'e.
\end{proof}

The following result provides a sufficient condition for (B6).
\begin{proposition}\label{prop:l0r0}
Let assumptions (B1) to (B5) to hold. Suppose that there is $l_0\in \mathcal{L}^1(Z;C(X))$ with $|\psi|\le l_0$, $\alpha\ge 0$ such that $\alpha \beta <1$, and
\[
\int_Z \max_{y\in \Gamma(x,z)} l_0(y,z')Q_z(dz') \le \alpha l_0(x,z),
\]
for all $x\in X$, $z\in Z$.
Then (B6) holds, with $R_0(K,z)= \frac 1{1-\alpha\beta} p_{K,z}(l_{0})$.
\end{proposition}
\begin{proof}
Choose $l_{t} = (\alpha\beta)^t l_0$, for $t=0,1,\ldots$. Then
\begin{align*}
\beta \int_Z \max_{y\in \Gamma(x,z)} l_t(y,z')Q_z(dz')  &=\beta  (\alpha\beta)^t\int_Z  \max_{y\in \Gamma(x,z)}l_0(y,z')Q_z(dz')\\
& \le  (\alpha\beta)^{t+1}l_0(x,z) = l_{t+1}(x_0,z_0).
\end{align*}
Hence, $w(x_0,z_0)=\frac 1{1-\alpha\beta} l_0(x_0,z_0)$ and $R_0(K,z)=\frac 1{1-\alpha\beta}p_{K,z}(l_0)$, for $K\in\mathcal{K}$ and $z\in Z$.
\end{proof}

\section{Applications}\label{Sect:Appl}

\subsection{Endogenous growth}\label{Sect:Endgrowth}
Endogenous growth models have become fundamental to understand economic growth. From the huge literature studying this field, few contributions consider an unbounded shock space. Some exceptions are \cite{Stach} and \cite{Kami}, but with uncorrelated shocks. I consider here the stochastic endogenous growth model studied in \cite{JMSS}, which is described as follows. The preferences of the agent over random consumptions sequences are given by
\begin{equation}\label{end0}
\max\  \E\sum_{t=0}^{\infty} \beta^t\frac{c_ t^{1-\sigma} \upsilon (\ell_t)}{1-\sigma},
\end{equation}
subject to
\begin{align}
\label{end1}&c_t + k_{t+1} + h_{t+1} \le z_t A k_t^{\alpha}(n_th_t)^{1-\alpha} + (1-\delta_k) k_t +  (1-\delta_h) h_t,\\
\label{end2}&\ell_t + n_t \le 1,\\
\label{end3}&c_t,k_t, h_t,\ell_t, n_t \ge 0
\end{align}
for all $t=0,1,\ldots$, with $k_0$ and $h_0$ given. Here, $\{z_t\}$ is a Markov stochastic process with transition probability $Q_z(\cdot)$ and $Z=[1,\infty)$; $c_t$ is consumption; $\ell_t$ is leisure; $n_t$ is hours spent working; $k_t$ and $h_t$ are the stock of physical and human capital, respectively; $\delta_k$ and $\delta_h$ are the depreciation rates on physical and human capital, respectively; and $\upsilon$ is a continuous function on $(0,1]$, strictly increasing. The usual non-negativity constraints on consumption, investment, leisure and hours worked apply. The feasible correspondence is thus
\begin{align*}
\Gamma(k,h,z)=\Big\{(k',h',c,n,\ell)\,:\, \mbox{\eqref{end1}--\eqref{end3} hold with $x'=x_{t+1}$, $x=x_t$}&\\
\mbox{ for $x=k,h,c,n,\ell,z$}\Big\}&,
\end{align*}
and the utility function is $U(c,\ell)= \frac{c^{1-\sigma} \upsilon (\ell)}{1-\sigma}$. Regarding the function $\upsilon$, we consider $\upsilon(\ell) = \ell ^{\psi (1-\sigma)}$.
The endogenous state space is $X=\R_+\times \R_+$ and the family of compact sets $\mathcal K$ is formed by compact sets in the product space $\R_+\times \R_+$.
The Markov chain is given by the log--log process
\begin{equation}\label{lnz}
\ln{z_{t+1}} = \rho \ln{z_t} + \ln{w_{t+1}},
\end{equation}
with $\rho\ge 0$ and where the $w$'s are i.i.d., with support in $W\subseteq [1,\infty)$. Let $\mu$ be the distribution measure\footnote{With correlated shocks, the method developed in \cite{MN} would require $\mu(z'\in Z_{j+1}|Z_j=z)=1$ for a suitable increasing family $\{Z_j\}_{j=1}^{\infty}$ of compact sets that fills $Z$. We do not impose this strong constraint on $\mu$. In fact, \eqref{lnz} do nos satisfy it if $W$ is unbounded and $\mu$ has not compact support.} of the $w$'s. Note that $\rho = 0$ corresponds to shocks $z_t$ that are i.i.d.. \cite{JMSS} suppose that $z_t=\exp{\left(\zeta_t-\frac{\sigma_{\epsilon}^2}{2(1-\rho^ 2)}\right)}$, where $\zeta_{t+1} = \rho \zeta_t + \epsilon_{t+1}$ and the $\epsilon$'s are i.i.d., normal with mean 0 and variance $\sigma^2_{\epsilon}$. This corresponds to \eqref{lnz} with\footnote{Since, from \eqref{lnz}, $z_{t+1} = z_t ^{\rho} w_{t+1}$, it is clear that, to keep $z\ge 1$, it is necessary (and sufficient) to have $w\ge 1$. Thus, the assumption that the random variable $\epsilon $ is normally distributed with mean 0 should be modified to fulfill the requirement that the random variable $w$ has support $W$ in $[1,\infty)$. The assumption $z\ge 1$ is usually imposed in growth models with a multiplicative structure, see \cite{SLP}.} $w_{t+1} = \exp{\left(\epsilon_{t+1} - \frac {\sigma_{\epsilon}^2}{2(1+\rho)}\right)}$.
We do not need to restrict $\epsilon$ to be normally distributed.
To shorten notation, let us define $\delta = \min\{\delta_k,\delta_h\}$, $\gamma = A \alpha^{\alpha} (1-\alpha)^{1-\alpha} + (1-\delta)$, and $g(k,h) = A k^{\alpha}(nh)^{1-\alpha} + (1-\delta) (k + h)$. Also, let $\Theta = \E{(w^{\frac{1-\sigma}{1-\rho}})} = \int_W w^{\frac{1-\sigma}{1-\rho}} \mu(dw)$.
\begin{theorem}\label{prop:end_bounded}
Consider the endogenous growth model described in \eqref{end0}--\eqref{lnz} with $0\le \sigma <1$ and $0\le \rho <1$.
If
\begin{equation}\label{OptG_condition}
\beta \gamma^{1-\sigma} \Theta  <1,
\end{equation}
then the associated Bellman equation admits a unique solution, $v^*$, in the set $V(0,R_0)$, where, for $K\in \mathcal K$ and $z\in Z$
\[
R_0(K,z) = \left(\frac{\Theta}{1-\beta\gamma^{1-\sigma} \Theta }\right)  z^{\frac{\rho(1-\sigma)}{1-\rho}} \,\max_{(k,h)\in K} g(k,h)^{1-\sigma}.
\]
Moreover, $v^*$ is the value function $v$ and
$p_{K,z}(T^nv_0 - v)$ converges to 0 as $n\to \infty$, for all $K\in \mathcal K$, $z\in Z$ and all initial guess $v_0\in V(0,R_0)$.
\end{theorem}
\begin{proof} We check all the hypotheses of Theorem \ref{GeneralDynPro}. It is clear that (B1)--(B5) are fulfilled. Regarding (B6), we will use Proposition \ref{prop:l0r0} for a suitable function $l_0$. Since $0\le \sigma <1$, both $U$ and $\upsilon$ are bounded from below by zero, and $\upsilon$ is bounded above by 1.
Since $z\ge 1$ and by the definition of $\delta$, we have
$
 zA k^{\alpha}(nh)^{1-\alpha} + (1-\delta_k) k + (1-\delta_h)h\le zg(h,k)
$.
Then
\[
\psi(k,h,z) \le \frac 1{1-\sigma}z^{1-\sigma}g(k,h)^{1-\sigma}\le  \frac 1{1-\sigma} z^{\frac{1-\sigma}{1-\rho}} g(k,h)^{1-\sigma} = l_0(k,h,z).
\]
Let us prove that $\beta \int_Z \widehat l_0(k,h,z,z')Q_z(dz')\le \alpha l_0(k,h,z)$, for all $(k,h)\in K$, for all $z\in Z$, and for all $K\in \mathcal{K}$, where $\alpha = \gamma^{1-\sigma} \Theta $. Here, to simplify notation in what follows, we have defined
\[
\widehat l_0(k,h,z,z')= \max_{(k',h',c,n,\ell)\in \Gamma(k,h,z)} l_0(k,h,z).
\]
First, we determine a bound for $\widehat l_0$.
To this end, consider the Lagrange problem
\begin{equation}\label{Lagrange}
\begin{aligned}
&\max \ g(k',h')^{1-\sigma},\\
&\mbox{s. t.: } k'+h' \le z g(k,h),\\
&\qquad k',h'\ge 0,
\end{aligned}
\end{equation}
and notice that its feasible set is larger than $\Gamma(k,h,z)$. The constraint is binding at the optimal solution, which is $k' = \alpha z g(k,h)$, $h'=(1-\alpha) z g(k,h)$. Substituting this into the objective function of \eqref{Lagrange}, we find its optimal value, $\gamma^{1-\sigma}  z^{1-\sigma} g(k,h)^{1-\sigma}$. Thus, $\widehat l_0(k,h,z,z') \le \frac 1{1-\sigma} (z')^{\frac{1-\sigma}{1-\rho}}\gamma^{1-\sigma}z^{1-\sigma} g(k,h)^{1-\sigma}$. Second, we use the conditional expectation $\int_Z (z')^{\frac{1-\sigma}{1-\rho}}Q_z(dz') = z^{\rho \frac{1-\sigma}{1-\rho}}\Theta $ to estimate
\begin{align*}
\int_Z \widehat l_0(k,h,z')Q_z(dz')& \le \gamma^{1-\sigma} z^{1-\sigma} z^{\rho \frac{1-\sigma}{1-\rho}}\Theta  g(k,h)^{1-\sigma}\\
& = \gamma^{1-\sigma}\Theta  l_0(k,h,z).
\end{align*}
Since $\beta \gamma^{1-\sigma} \Theta <1$, Proposition \ref{prop:l0r0} applies. The expression for $R_0$ requires a simple computation.
\end{proof}

\subsection{Asset Prices in an Exchange Economy}\label{sect:Exchange}
\cite{Lucas} studied the determination of equilibrium asset prices in a pure exchange economy in a framework that has become classical in the economics and financial literature. Boundedness of the utility function, as well as compactness of the space of shocks, are important hypotheses in the development of this model. In this section, we show that these hypotheses can be dispensed with by using the results of Theorem \ref{GeneralDynPro}. In this way, we significantly extend the model's range of applicability.

We closely follow \cite{SLP} in the exposition of the problem.
The preferences of the representative consumer over random consumption sequences are
\begin{equation}\label{ULucas}
\E\ \sum_{t=0}^{\infty} \beta ^t u(c_t),
\end{equation}
where $u:\R_+\longrightarrow \R_+$ is continuous, not necessarily bounded, with\footnote{If $u$ does not satisfy $u\ge 0$ and $u(0)=0$, but is bounded from below, it may be modified to $u(c)-u(0)$ to fulfill our hypotheses.} $u(0)=0$ and where $\beta\in (0,1)$. There are $i=1,\ldots,k$ productive assets taking values on a set $Z\subseteq \R^k_+$, not necessarily compact, with Borel sets $\mathcal{Z}$. The components $z_i$ of the vector $z=(z_1,\ldots,z_k)^{\top}$ in $Z$ represents the dividend paid by one unit of asset $i$. In the description of the model, all vectors are considered column vectors, and the symbol $^{\top}$ denotes transposition. We assume that the dividends follow a Markov process, with stationary transition function $Q$ on $(Z,\mathcal{Z})$. Assets are traded on a competitive stock market at an equilibrium price given by a stationary continuous price function $p:Z\longrightarrow \R^k_+$, where $p(z) = (p_1(z),\ldots,p_k(z))^{\top}$ is the vector of asset prices if the current state of the economy is $z$ (the notation for prices should not be confused with the notation for seminorms, which always carry a subindex). The goal is to characterize equilibrium asset prices. Let $x=(x_1,\ldots, x_k)^{\top}\in \R^k_+$ be the vector of the consumer's asset holdings. Given the price function $p$, the initial state of the economy $z_0$ and initial asset holdings $x_0$, the consumer chooses a sequence of plans for consumption and end-of-period asset holdings that maximizes discounted expected utility \eqref{ULucas} subject to
\begin{align}
&c_t + x^{\top} _{t+1}p(z_t) \le  x_t^{\top} (z_t + p(z_t))  \quad \mbox{for all $z^t$, for all $t$,}\label{Lconst1}\\
&c_t, x_{t+1} \ge 0\quad \mbox{for all $z^t$, for all $t$.}\label{Lconst2}
\end{align}
The consumer holds exactly one unit of each asset in equilibrium, hence we can restrict the state space to $X=[0,\overline x]^k$, with $\overline x >1$, with its Borel subsets $\mathcal{X}$. The correspondence $\Gamma:X\times Z \longrightarrow 2^X$ is
\[
\Gamma(x,z) =\{y\in X\,:\, y^{\top} p(z)  \le x^{\top} (z+p(z))\}.
\]
Assuming that $p$ is continuous, $\Gamma$ is nonempty, compact valued and continuous.
Given the price $p$, the dynamic programming equation is
\begin{equation}\label{LucasDP}
v(x,z) = \max_{y\in \Gamma(x,z)} \Big\{u(x^{\top} z  + (x-y)^{\top}p(z)) + \beta \int_Z v(y,z') Q_z(dz')\Big\}.
\end{equation}
We will look for solutions to this functional equation in the class $\mathcal{L}^1(Z;C(X))$. Since the state space is compact, and the utility function $u$ is bounded from below, we take the trivial family of compact sets $\mathcal{K}=\{X\}$ in this model, and not the whole family of compact subsets of $X$.
We impose the following assumptions.
\begin{align}
&u \mbox{ is nondecreasing and concave}\label{uLucas};\\
&p_i(z) \le a_i^{\top}  z + b_i \mbox{ for some vectors  $a_i\ge 0$ and scalars $b_i> 0$},\label{pLucas}\\
& \mbox{$i=1,\ldots,k$}\nonumber.
\end{align}
Hence, we look for equilibrium prices in the class of functions that are bounded by an affine function. Other possibilities could obviously be explored. We state two results about the existence of equilibrium in a Lucas asset pricing model satisfying \eqref{uLucas} and \eqref{pLucas} under two different regimes for the Markov chain.

\begin{enumerate}[{(}{M}1{)}]
\item
The Markov chain is given by $z_{t+1} =  Bz_t + w_t$, where $B$ a matrix of order $k$ with non--negative entries and norm\footnote{The norm of a matrix $B$ is defined by $\|B\| = \sup\left\{\frac{\|Bx\|}{\|x\|}\,:\, x\in \R^l\mbox{ with $x\neq 0$}\right\}$. The condition $\|B\|<1$ is equivalent to saying that the spectral radius of $B$---the maximum of the module of the eigenvalues of $B$---is less than one.} $\|B\|<1$, and where $\{w_t\}_{t=1}^{\infty}$ are i.i.d. random vectors with support in a Borel subset $W\subseteq \R_+^k$ with finite expectation, $0\le \E{w}<\infty$.
\item
The Markov chain is given by $z_{i,t+1} = z_{i,t} ^{\rho_i} w_{i,t+1}$, for all $t=0,1,2,\ldots$, where $0\le \rho_i \le 1$ for all $i=1,\ldots,k$, and where $\{w_t\}_{t=1}^{\infty}$ are i.i.d. random vectors\footnote{Hence, we are now considering a linear log--log system of uncoupled equations for the evolution of dividends.} with support in a Borel subset $W\subseteq [1,\infty)^k$ such that
\begin{equation}\label{mulLucas}
\rho_i \le 1 \mbox{ for all } i=1,\ldots,k \mbox{ and, if $\rho_i=1$, then }\E{w_i} < 1/\beta.
\end{equation}
\end{enumerate}

In the proof that follows, as well as in the rest of the paper, we use the same notation for inequalities between scalars and inequalities between vectors, which have to be understood in a pointwise manner.

\begin{theorem}\label{th:Lucas1}
Consider the Lucas Asset Pricing model described above in \eqref{ULucas}-\eqref{Lconst2}, for which \eqref{uLucas} and \eqref{pLucas} hold and the Markov chain satisfies either (M1) or (M2). Then there is a unique solution of \eqref{LucasDP} in $V(0, R_0)$ for a suitable $R_0$, which is the value function $v$ of the problem, and the conclusions of Theorem \ref{GeneralDynPro} hold.
\end{theorem}
\begin{proof} Let $A$ be the matrix whose columns are the vectors $a_1,\ldots,a_k$ in \eqref{pLucas}, and let $b=(b_1,\ldots,b_k)^{\top}$, hence we can write $0\le p(z)\le Az + b$.
Let us construct a family of functions $\{l_t\}_{t=0}^{\infty}$ satisfying assumption (B6). Obviously
\[
x^{\top} (z+p(z)) \le \overline x ^{\top} z(I_k + A)  + \overline x ^{\top} b,
\]
where $I_k$ is the indentity matrix. Since $u$ is increasing and concave, for a supergradient $\overline u$ of $u$ at $\overline x ^{\top} b>0$, we have
\[
\psi (x,z) = u(x^{\top} (z+p(z)) \le u(\overline x^{\top} b) + \overline u\; \overline x^{\top} (I_k+A)z.
\]
We define $l_0(z) =  u(\overline x^{\top} b) + \overline u\; \overline x^{\top} (I_k+A) z$ and,  recursively, $l_{t+1}(z) = \beta \int_Z l_t(z')Q(z,dz')$, for $t=0,1,\ldots$.

Suppose first that $Q$ satisfies (M1). In this case, $\E_{z}{z'} = Bz + \E {w}$, hence
\begin{align*}
l_1(z) &= \beta \int_Z l_0(z')  Q(z,dz')\\
& =   \beta  u(\overline x^{\top} b) + \beta \overline u\; \overline x^{\top}  (I_k+A) \E_{z}{z'}\\
& = \beta (  u(\overline x^{\top} b) + \overline u\; \overline x^{\top} (I_k+A) (Bz+\E{w}).
\end{align*}
We will prove by induction that
\begin{align*}
l_{t}(z) &= \beta^{t} u(\overline x^{\top} b) + \beta^{t}  \overline u\; \overline x^{\top} (I_k+A) \big(B^tz + (B^{t-1}+\cdots +  I_k )\E{w}\big),
\end{align*}
for all $t=1,2,\ldots$. For $t=1$ it has been just proved. Suppose it is true for $t$. Then
\begin{align*}
l_{t+1} &= \beta^{t+1} u(\overline x^{\top} b) +  \beta^{t+1}  \overline u\; \overline x^{\top} (I_k+A) \big(B^t\E_z{z'} + (B^{t-1}+\cdots + I_k )\E{w}\big)\\
&=\beta^{t+1} u(\overline x^{\top} b) \\
&\quad + \beta^{t+1}  \overline u\; \overline x^{\top} (I_k+A) \big(B^{t}(Bz+\E{w})
+ (B^{t-1}+\cdots + I_k )\E{w}\big),
\end{align*}
and we are done. On the other hand, $(B^{t-1}+\cdots +  I_k )\E{w}\le  (I_k -B)^{-1}\E{w}$, since $B$ has nonnegative entries, $\|B\|<1$ and $\E{w}>0$. Hence, the series
$w_0(z) = \sum_{t=0}^{\infty}l_t(z)$ is unconditionally convergent, since it is bounded by the function of  $L^1(Z)$ defined by
\[
\overline w_0(z) := \frac 1{1-\beta} (u(\overline x^{\top} b)  + \beta \overline u\; \overline x^{\top}(I_k+A)  \big((I_k - \beta B)^{-1} z + (I_k - B) ^{-1} \E{w} \big),
\]
where we have used $\sum_{t=0}^{\infty} (\beta B)^t = ( I_k - \beta B)^{-1}$.
Hence (B6) holds with $R_0(z)=\sum_{t=0}^{\infty}p_z(l_t)$.

If $Q$ satisfies (M2), then $\E_z{z'}= (z_1^{\rho_1}\E{w_1},\ldots,z_k^{\rho_k}\E{w_k})^{\top}$. Define, as above, $l_0(z) =  u(\overline x\cdot b) + \overline u\; \overline x\cdot z (I_k+A)$ and let $l_{t+1}(z) = \beta \int_Z l_t(z')Q(z,dz')$, for $t=0,1,\ldots$. Then it is easy to prove by induction that
\[
l_{t} (z) = \beta^{t} u(\overline x^{\top} b) + \beta^{t}  \overline u\; \overline x^{\top} (I_k+A) \left(z_1^{\rho_1^t}\Pi_{s=0}^{t-1}\E{\left(w_1^{\rho_1^s}\right)},\ldots,z_k^{\rho_k^t}\Pi_{s=0}^{t-1}\E{\left(w_k^{\rho_k^s}\right)}\right)^{\top}.
\]
By Jensen's inequality, $\E{\left(w_i^{\rho_i^s}\right)} \le (\E{w_i})^{\rho_i^s}$ and thus
\[
l_{t}(z) \le \beta^{t} u(\overline x^{\top} b) + \beta^{t}  \overline u\; \overline x^{\top} (I_k+A) \left(z_1^{\rho_1^t}\left(\E{w_1}\right)^{1/(1-\rho_1)},\ldots, z_k^{\rho_k^t}\left(\E{w_k}\right)^{1/(1-\rho_k)}\right)^{\top}.
\]
In the case that $\rho_i<1$ for all $i=1,\ldots, k$, the series $w_0(z) = \sum_{t=0}^{\infty} l_t(z)$ is clearly (unconditionally) convergent since $\beta <1$ and $\rho_i<1$ for all $i=1,\ldots,k$. The ratio test can be used to prove this claim. In the case in which $\rho_j=1$ for some $j$, then the bound above no longer applies, as a term $z_j (\E{w_j})^t$ appears in position $j$ of the vector $\left(z_1^{\rho_1^t}\Pi_{s=0}^{t-1}\E{w_1^{\rho_1^s}},\ldots,z_k^{\rho_k^t}\Pi_{s=0}^{t-1}\E{w_k^{\rho_k^s}}\right)^{\top}$. However, the assumption $\beta \E{w_j}<1$ assures convergence of the series $w_0(z) = \sum_{t=0}^{\infty} l_t(z)$.

Hence, in both cases considered, (M1) and (M2), the condition \eqref{mulLucas} guarantees that (B6) holds with $R_0(z)=\sum_{t=0}^{\infty}p_z(l_t)$. Thus, Theorem \ref{GeneralDynPro} applies. 
\end{proof}

%

To complete the circle, we have to prove that our conjecture \eqref{pLucas} about the equilibrium price holds. Following \cite{Lucas} or \cite{SLP}, we now assume the further conditions:
\begin{align}
\label{eqLucas1}
&\mbox{$u(0)=0$, $u$ is continuously differentiable, with $u'(c)>0$ for all $c\ge 0$,}\\
&\mbox{and strictly concave;}\nonumber
\end{align}
we also impose
\begin{align}
\label{eqLucas3}
&\mbox{there are constants $\gamma,\delta \ge 0$ such that } cu'(c) \le \gamma c + \delta,  \mbox{ for all $c\ge 0$;}\\
& \mbox{there exists $a>0$ such that }u'(\mathbf{1}^{\top}z) \ge a \mbox{ for all $z\in Z$, where $\mathbf{1}=(1,\ldots,1)^{\top}$}.\label{eqLucas4}
\end{align}

A function like $u(c) = c^{1-\sigma}/(1-\sigma) + c$, with $0\le \sigma <1$, satisfies \eqref{eqLucas1}-\eqref{eqLucas4}. Also, if $Z$ is bounded, then \eqref{eqLucas1} implies \eqref{eqLucas4}.

Finding an equilibrium price function $p(z)=(p_1(z),\ldots,p_k(z))^{\top}$ is equivalent to finding functions $\phi_1(z),\ldots, \phi_k(z)$ that satisfy the $k$ independent functional equations
\begin{equation}\label{Lucasphi}
\phi_i(z) = h_i(z) + \beta \int _Z \phi_i(z') Q(z,dz'), \quad i=1,\ldots,k,
\end{equation}
where $h_i(z)=\beta \int_Z z'_i u'(\mathbf{1}^{\top}z')  Q(z,dz')$, for all $i =1,\ldots,k$.
\cite{Lucas} shows that a solution to \eqref{Lucasphi} provides an equilibrium price $p$ given by
\begin{equation}\label{priceLucas}
p_i(z) = \frac{\phi_i(z)}{u'(\mathbf{1}^{\top} z )},\quad \mbox{for $i=1,\ldots,k$}.
\end{equation}
Let, as in \cite{Lucas}, the operator $T_i$ be
\[
T_if (z) = h_i(z) + \beta \int _Z f(z') Q(z,dz'),\quad \mbox{for all } f\in L^1(Z,\mathcal{Z},Q_z),\quad  i=1,\ldots,k.
\]
Note that, in this context, the seminorms are simply defined by $p_z(f) = \int_Z |f(z')|Q_z(dz')|$. It is pretty clear that the COP associated to $T_i$ is given by
\[
Lp (z) = \beta \int_Z p(z')Q_z(dz'),
\]
where $p$ belongs to a suitable set $C$ as defined in \eqref{SetC}.

\begin{theorem}\label{prop:Lucas1}
Consider the Lucas Asset Pricing model described above in \eqref{ULucas}-\eqref{Lconst2}, for which \eqref{eqLucas1}--\eqref{eqLucas4} hold and the Markov chain satisfies either (M1) or (M2). Then there is an equilibrium price $p$ satisfying \eqref{pLucas}.
\end{theorem}
\begin{proof}
Note that $z\ge 0$ and \eqref{eqLucas3} imply
\begin{equation}\label{hineq}
z_i u'(\mathbf{1}^{\top} z)  \le (\mathbf{1}^{\top} z) u'(\mathbf{1}^{\top} z)   \le \gamma (\mathbf{1}^{\top} z) + \delta.
\end{equation}
Suppose that $Q$ satisfies (M1). Then
\begin{align*}
h_i(z) &\le \beta  \left(\gamma \left(\mathbf{1}^{\top} \E_z{z'}\right) + \delta\right)
\le
\beta  \left(\gamma \left(\mathbf{1}^{\top} (Bz+\E{w})\right) + \delta\right).
\end{align*}
This implies that the operator $T_i$ is a self--map in $L^1(Z)$, for all $i=1,\ldots,k$.  We want to apply Theorem \ref{GeneralDynPro} to each of the operators $T_i$, where the COP associated to $T_i$ is given just above the theorem.
Let $l_0(z)=h_i(z)$ and define 
\[
l_t(z) = \beta^t  \left(\gamma \left(\mathbf{1}^{\top} (B^t z+\E{w})\right) + \delta\right),
\]
for $ t=1,2,\ldots$.
It is immediate to check that $l_{t+1}\ge \beta \int_Z l_t(z')Q(z,dz')$ and that the series $w_0(z)=\sum_{t=0}^{\infty} l_t(z)$ is (unconditionally) convergent, since $\beta \|B\| <1$. The sum of this series is $w_0(z)=\gamma \mathbf{1}^{\top} (I_k-\beta B)^{-1}z  + \frac 1{1-\beta}(\gamma \mathbf{1}^{\top} \E w + \delta)$. Hence, (B6) holds. Moreover, following analogous reasonings as in the proof of part (b) of Theorem \ref{GeneralDynPro}, the fixed point of $T_i$, $\phi_i$, satisfies $\phi_i\le w_0$, and thus, by \eqref{eqLucas4}
\[
p_i(z) = \frac{\phi_i(z)}{ u'(\mathbf{1}^{\top} z)} \le \frac 1a \gamma \mathbf{1} ^{\top} (I_k-\beta B)^{-1} z + \frac 1{a(1-\beta)}(\gamma \mathbf{1}^{\top} \E w  + \delta),
\]
for all $i=1,\ldots,k$, where the right hand side is an affine function of $z$. Then, $p=(p_1,\ldots,p_k)^{\top}$ satisfies \eqref{pLucas} with
\[
a_i^{\top} = \frac 1a \gamma \mathbf{1} ^{\top} (I_k-\beta B)^{-1},\quad b_i = \frac 1{a(1-\beta)}(\gamma \mathbf{1}^{\top} \E w  + \delta),\mbox{ for all $i=1,\ldots,k$}.
\]

Suppose that $Q$ satisfies (M2). Now $\E_z{z'}= (z_1^{\rho_1}\E{w_1},\ldots,z_k^{\rho_k}\E{w_k})^{\top}$. From \eqref{hineq}, we have
\[
h_i(z)\le \beta(\gamma (\mathbf{1}^{\top} (z_1^{\rho_1}\E{w_1},\ldots,z_k^{\rho_k}\E{w_k}))+\delta).
\]
Let $l_0(z)=h_i(z)$ and $l_{t+1}(z) = \beta \int_Z l_t(z')Q(z,dz')$, for $t=0,1,\ldots$.
When $0\le \rho_i<1$ for all $i=1,\ldots,k$, using similar arguments as in the proof of Theorem \ref{th:Lucas1}, we have
\begin{align*}
l_t(z) &\le \beta^t \gamma \mathbf{1}^{\top} \left(z_1^{\rho_1^t} (\E{w_1})^{1/(1-\rho_1)},\ldots, z_k^{\rho_k^t}(\E{w_k})^{1/(1-\rho_k)}\right)+\beta^t\delta \\
&\le \beta^t (\gamma \mu \mathbf{1}^{\top} z + \delta),
\end{align*}
where $\mu:=\max\left\{z_1(\E{w_1})^{1/(1-\rho_1)},\ldots, z_k(\E{w_k})^{1/(1-\rho_k)}\right\}$.
Hence, the infinite series $w_0(z)=\sum_{t=0}^{\infty} l_t(z)=\frac {\gamma \mu \mathbf{1}^{\top} z + \delta}{1-\beta} $ is (unconditionally) convergent, (B6) holds and the fixed point of $T_i$, $\phi_i$, satisfies $\phi_i\le w_0$. It is clear then that the price $p=(p_1,\ldots,p_k)^{\top}$ defined in \eqref{priceLucas} satisfies \eqref{pLucas} with  $a_i = \frac {\gamma \mu}{a(1-\beta)}$ and $b_i = \frac{\delta}{a(1-\beta)}$, for all $i=1,\ldots,k$. In the case in which some $\rho_j=1$, the coordinate $j$ on the vector  $\left(z_1^{\rho_1^t}\Pi_{s=0}^{t-1}\E{(w_1^{\rho_1^s})},\ldots,z_k^{\rho_k^t}\Pi_{s=0}^{t-1}\E{(w_k^{\rho_k^s})}\right)$ is equal to $z_j \E{w_j}$, and then $\beta \E{w_j}<1$ is required to have convergence of the series $\sum_{t=0}^{\infty} l_t(z)$, which is then bounded by an affine expression in $z$; hence, as in the previous case, \eqref{pLucas} holds.
\end{proof}

\section{Conclusions}\label{Sect:Conclusions}
In this paper, we develop a general framework to analyze stochastic dynamic problems with unbounded utility functions and unbounded shock space. We obtain new results concerning the existence and uniqueness of solutions to the Bellman equation through a fixed point theorem that generalizes the results known for Banach contractions and local contractions. This generalization is possible by considering seminorms that give a different treatment to the endogenous state variable and the exogenous one. While a supremum norm on arbitrary compact sets is considered in the former variable, an $L^1$ type norm is in the latter variable. Putting together this definition with the aforementioned generalization of the local contraction concept, we are able to maintain the monotonicity (in a mild sense) of the Bellman operator, thus proving that it is essentially a contractive operator. The usefulness of the approach and the applicability of the results are clearly revealed in the analysis of two fundamental models of economic analysis: an endogenous growth model with a multiplicative structure in the shocks and the Lucas model of an exchange economy. The combination of unbounded rewards and unbounded shocks makes it hard to prove the existence of a unique fixed point of the Bellman equation. In this sense, another benefit of the paper is to provide a secure method to check the hypotheses needed to apply the approach, based on assumption (B6), and one that can be used straightforwardly to analyze other models. A challenging problem is to extend the theorems to deal with the unbounded from below case in a more satisfactory way, as done in \cite{RZRP} or \cite{Vailakis} for the deterministic case, by introducing a suitable family of pseudodistances.

\newpage

\appendix

\section{Proofs of auxiliary results}\label{App:Proofs}

\textsc{Proof of Lemma \ref{B_defined}}.
Due to the subhomogeneity of $L$ for finite sums, $L(r_0+Lr_0+\cdots+L^Tr_0)\le Lr_0+\cdots +L^{T+1}r_0\le R_0$, for all finite $T$. Letting $T\to \infty$, we obtain $r_0+LR_0\le R_0$. Let $x\in V_F(x_0,R_0)$, so $d_a(x_0,x)\le R_0(a)$ for all $a\in A$. By the triangle inequality and since $T$ is an $L$--local contraction
\begin{align*}
d_a(x_0,Tx) &\le d_a(x_0,Tx_0) + d_a(Tx_0,Tx)\\
& \le d_0(a) + (Ld_a) (x_0,x) \\
&\le d_0(a) + (LR_0)(a)\\
&\le R_0(a).
\end{align*}
This proves (a). To show (b), note that, by the same arguments used to prove (a), for $L^tR_0 \le L^tr_0+L^{t+1}r_0+\cdots$, for all $t=0,1,\ldots$. Then $L^tR_0(a)$ is bounded by the remainder of the convergent series $R_0(a)$,  thus it converges to 0 as $t\to\infty$, for all $a\in A$.\hfill{\it Q.E.D.}

\strut


Given $f\in \ca{X\times Z}$, we denote
\[
\widehat f (x,z) := \max_{y \in \Gamma(x,z)} f(y,z),\quad \widehat{|f|} (x,z) := \max_{y \in \Gamma(x,z)} |f(y,z)|.
\] 

We will make use of the following lemma in the main text and along this appendix.
\begin{lemma}\label{welldefined}

(1)
For all $f\in \ca{X\times Z}$, both $\widehat f$, $\widehat{| f|}\in \ca{X\times Z}$.

(2)
For all $f\in \mathcal{L}^1(Z;C(X))$, both $\widehat f$, $\widehat{| f|} \in L^1(Z,\mathcal{Z},Q_z)$, for all $z\in Z$.

\end{lemma}
\begin{proof}
(1)
Given the assumption made about the continuity of $\Gamma$, by the Berg\'e Theorem of the Maximum, the map $x\mapsto \widehat f (x,z)$ is continuous, for any $z\in Z$ fixed, and by the Measurable Theorem of the Maximum, $z\mapsto \widehat f (x,z)$ is Borel measurable; thus, $\widehat f$ is a Carath\'eodory function on $X\times Z$. Obviously, the same is true for $\widehat{|f|}$.

(2)
Since $p_{K,z}(f) <\infty$ for all $K\in \mathcal{K}$ and all $z\in Z$, and $\Gamma(x,z)$ is a compact set for any $x\in X$, $z\in Z$, then $\int_Z \widehat{|f|}(y,z') Q_z(dz') = p_{\Gamma(x,z),z} (f) <\infty$. Obviously, the same is true for $\widehat{f}$.
\end{proof}

\noindent
\textsc{Proof of Lemma \ref{Preparation_Th}}. We organize the proof in several previous lemmas.
\begin{lemma}\label{rt_family}
Let assumptions (B1) to (B6) to hold. Then
\begin{enumerate}
\item $\sum_{t=0}^{\infty} L^t p^{l_0} <\infty$;
\item $R_0[\Gamma]\in \mathcal{L}^1(Z;C(X))$ and $ p^{l_0} + LR_0\le R_0$.
\end{enumerate}
\end{lemma}
\begin{proof}
Given $x\in X$ and $z\in Z$, $\beta p^{l_t}[\Gamma](x,z) =  \beta \int_Z \max_{y\in \Gamma(x,z)} l_t(y,z')Q_z(dz') \le l_{t+1}(x,z)$, hence $p^{l_t}[\Gamma]\in \mathcal{L}^1(Z,C(X))$ and then $L p^{l_t}(K,z)=\beta p_{K,z}(p^{l_t}[\Gamma])\le p_{K,z}(l_{t+1})$, for all $t=0,1,\ldots$. Thus, $L^tp^{l_0} \le L^{t-1}p^{l_1}\le \cdots \le p^{l_t}$. By (B6), the series $\sum _{t=0}^{\infty} p^{l_t}(K,z)$ converges for all $K\in \mathcal{K}$ and $z\in Z$, thus $\sum_{t=0}^{\infty} L^t p^{l_0}$ converges.
To conclude the proof, by the triangle inequality
\begin{align*}
 p_{K,z} (p^{l_0}[\Gamma] + \cdots + p^{l_t}[\Gamma]) &\le  p_{K,z} (p^{l_0}[\Gamma]) + \cdots +  p_{K,z}(p^{l_0}[\Gamma])\\
&\le (p_{K,z}(l_1) +\cdots +p_{K,z}(l_{t+1}).
\end{align*}
Letting $t\to \infty$ and adding $p_{K,z}(\psi) $ to both sides of the above inequality, we have $ p_{K,z}(\psi) + p_{K,z}( R_0[\Gamma]) \le R_0(K,z)$, showing at the same time that $R_0[\Gamma]\in \mathcal{L}^1(Z;C(X))$.
\end{proof}

\begin{lemma}\label{Tf_well}
Let assumptions (B1) to (B6) to hold. Then $f\in V(0,R_0)$ implies $Tf\in \mathcal{L}^1(Z;C(X))$.
\end{lemma}
\begin{proof}
Let $f\in \mathcal{L}^1(Z;C(X))$. We use the notation $f_x$ and $f^z$, whose meaning is explained in Appendix \ref{App:FS}. The function $f_x$ is Borel measurable for all $x\in X$ and $Q_z$--integrable for any $z\in Z$. Thus, $f_x$ can be written as the difference of two positive, $Q_z$--integrable functions, $f_x=f_x^+ - f_x^-$, where $f_x^+=\max(f_x,0)$ and $f_x^-=\max(-f_x,0)$. Applying Theorem 8.1 in \cite{SLP}, both $Mf_x^+$ and $Mf_x^-$ are Borel measurable. Since $(Mf)_x = M(f_x) = M(f_x^+) - M(f_x^-)$, $(Mf)_x$ is measurable for any $x\in X$.
To see that $(Mf)^z$ is continuous, consider a sequence $\{x_n\}$ in $X$ that converges to $x\in X$. Then the sequence and its limit form the compact set $K=\{x_n\}\cup\{x\}$. Let $f_n:=f_{x_n}$, for $n\ge 1$. For all $z'\in Z$, $f_n(z')\to f_x(z')$ as $n\to \infty$, since $f$ is continuous in $x$. Moreover, $|f^{z'}|\le \sup_{x\in K} |f^{z'}(x)|$, and $z'\mapsto \sup_{x\in K} |f^{z'}(x)|$ is $Q_z$--integrable by definition of $\mathcal{L}^1(Z;C(X))$, thus by the Lebesgue dominated convergence theorem
\[
(Mf)(x_n,z) = \int _Z f_n(z') Q_z(dz') \to \int _Z f_x(z') Q_z(dz') = (Mf)(x,z),
\]
thus $(Mf)^z$ is continuous. Hence, $Mf$ is a Carath\'eodory function and thus $U(x,y,z) + \beta Mf(y,z)$ is continuous in $(x,y)$ for all $z$, and it is Borel measurable in $z$ for all $(x,y)$. By the Berg\'e Maximum Theorem, the function $Tf$ is thus continuous in $x$ for all $z$, and by the Measurable Maximum Theorem, it is Borel measurable for any $x$. In short, the function
\[
(x,z)\mapsto Tf(x,z) = \max_{y\in \Gamma(x,z)} (U(x,y,z) + \beta Mf(y,z)),
\]
is a Carath\'eodory function. Moreover, if $f\in F$ and $x\in X$, $z\in Z$
\begin{align*}
|Tf(x,z)|&\le
|\max_{y\in \Gamma(x,z)}  U(x,y,z)|+ \beta \max_{y\in \Gamma(x,z)} {\displaystyle \int_Z}\max_{y\in \Gamma(x,z)}  |f(y,z')|Q_z(dz')\\
& \le l_0(x,z) + \beta \int _Z \max_{y\in \Gamma(x,z)} w(y,z') Q_z(dz')\\
& \le l_0(x,z) + \beta p_{\Gamma(x,z),z} (f).
\end{align*}
Since $\Gamma(x,z)\in \mathcal{K}$, for $f\in V(0,R_0)$, we have $p_{\Gamma(x,z),z} (f) \le 
R_0[\Gamma](x,z)$. By Lemma \ref{rt_family}, $p_{K,z}(l_0) + \beta p_{K,z}(R_0[\Gamma])\le R_0(K,z)$. Hence 
$
p_{K,z} (Tf) \le R_0(K,z)$.
This proves that $Tf\in V(0,R_0)$, and hence that $Tf \in \mathcal{L}^1(Z,C(X))$.
\end{proof}

\begin{lemma}\label{lemma:C}
Let assumptions (B1) to (B6) to hold. Then $D(V(0,R_0))\subseteq C$. 
\end{lemma}
\begin{proof}
Since $f\in V(0,R_0)$, $p^f\le  R_0$, hence we can take $c=1$. Also, $p^f\in \ca{X\times Z}$, since $p^f[\Gamma] (x,z) = \int_Z \max_{y\in \Gamma(x,z)} |f(y,z')|Q_z(dz')$ is continuous in $x$ and Borel measurable in $z$, by Lemma \ref{welldefined}. Moreover, $p^f[\Gamma] \le R_0[\Gamma]$ implies $ p_{K,z} (p^f[\Gamma]) \le   p_{K,z}(R_0[\Gamma]) \le \frac 1{\beta} R_0(K,z)$, by Lemma \ref{rt_family}. Hence, $p^f[\Gamma]\in \mathcal{L}^1(Z,C(X))$.
\end{proof}

Now, we are in position to prove Lemma \ref{Preparation_Th}.
First, let us see
 that $L:C\longrightarrow C$. Let $p\in C$; by the definition of the operator $L$ and Lemma \ref{rt_family}
\[
Lp (K,z) =  \beta p_{K,z} (p[\Gamma])\le \beta p_{K,z} (cR_0[\Gamma])\le  c R_1(K,z)\le  c R_0(K,z),
\]
and so, $Lp[\Gamma]  \le c R_0[\Gamma]$ and $Lp[\Gamma] \in \mathcal{L}^1(Z,C(X))$.
Second,
we prove that the assumptions (I) to (VI) are fulfilled. Regarding (I), note that $p+q\in C$ if $p,q\in C$, trivially, as well it is also immediate that if $p'\in C$ and $p\le p'$, then $p\in C$. On the other hand, if a countable chain of partial sums $p_0, p_0+p_1, p_0+p_1+p_2, \ldots,$ is bounded by an element $P$ in $C$, then the infinite sum, $p:=\sum _{n=0}^{\infty} p_n$, is well defined and $p\le P\le cR_0$ for some constant $c$. Moreover, since $p[\Gamma]\le c R_0[\Gamma]$ and $R_0[\Gamma]\in \mathcal{L}^1(Z;C(X))$ by Lemma \ref{rt_family}, the Monotone Convergence Theorem implies that $p[\Gamma](x,\cdot)$ is $Q_z$--integrable for all $z\in Z$ and all $x\in X$. On the other hand, each function $p_i[\Gamma](\cdot, z)$ is continuous in $x$, for all $i=1,2,\ldots$. By the Wierstrass M test, the function $p[\Gamma](\cdot, z)$ is also continuous in $x$ for all $z\in Z$. These two observations imply that $p[\Gamma]\in \mathcal{L}^1(Z;C(X))$.
(II) is trivial; (III) holds, since the integral is monotone, and regarding (IV), it holds true, since for all $p,q\in C$, $p_{K,z}(p[\Gamma]+q[\Gamma]) \le p_{K,z}(p[\Gamma])+p_{K,z} (q[\Gamma])$ by definition of the seminorms $p_{K,z}$, hence
\begin{align*}
L(p+q)(K,z) &= p_{K,z}(p[\Gamma]+q[\Gamma])\\
& \le p_{K,z}(p[\Gamma])+p_{K,z} (q[\Gamma])\\
&=Lp(K,z) + Lq(K,z).
\end{align*}
$L$ is clearly sup-preserving in $C$ by the Monotone Convergence Theorem, hence (V) also holds.
Finally, (VI) is implied by Lemma \ref{rt_family} and Lemma \ref{Tf_well}.

\section{Function space}\label{App:FS}
We describe in this section the function space used in Section \ref{Sect:BO} and we prove Lemma \ref{lemma:function_space}.

Let the measurable space $(Z,\mathcal{Z})$, where $Z$ is the space of shocks and $\mathcal{Z}$ is the $\sigma$--algebra of Borel of $Z$. Remember that $Q$ is a transition function $Q:Z\times \mathcal{Z}\rightarrow [0,1]$ satisfying
\begin{enumerate}
\item for each $z\in Z$, $Q(z,\cdot)$ is a probability measure on $(Z,\mathcal{Z})$; and
\item for each $A\in \mathcal{Z}$, $Q(\cdot,A)$ is a Borel measurable function.
\end{enumerate}

To simplify notation, let $Q_z=Q(z,\cdot)$. For each $z\in Z$, let $L^1(Z,\mathcal{Z},Q_z)$ be the space of Borel measurable and $Q_z$--integrable functions
 $g:Z\longrightarrow \R$. Let the  $L^1$ norm with respect to the fixed probability measure $Q_z$
\[
\|g\|_{z} = \int_Z |g(z')|Q_z(dz'),
\]
where the notation $\|\cdot\|_{z}$ means that integration is with respect to the probability measure $Q_z$. In what follows we will omit the $\sigma$--algebra $\mathcal{Z}$ from the notation.

\strut
A function $f:X\times Z\longrightarrow \R$ is a Carath\'eodory function on $X\times Z$ if it satisfies
\begin{enumerate}
\item for each $x\in X$, the function $f_x:=f(x,\cdot) : Z \longrightarrow \R$ is Borel measurable;
\item for each $z\in Z$, the function $f^z:=f(\cdot,z):X\longrightarrow \R$ is continuous.
\end{enumerate}
Under our assumptions, a Carath\'eodory function is jointly measurable in $X\times Z$, see \cite{AB}, Lemma 4.50. Also, a function that is Carath\'eodory on $X\times Z$ is obviously  Carath\'eodory on $A\times Z$ for all $A\subseteq X$. Let us denote by $\ca{A\times Z}$ the set of all Carath\'eodory functions on $A\times Z$.

\strut
Let $\mathcal{K}$ denote the family of compact subsets $K\subseteq X$. Given $z\in Z$, consider the  probability measure $Q_z(\cdot)$. For a Carath\'eodory function $f$ on $X\times Z$, let
\[
p_{K,z}(f) = \int_Z \max_{x\in K} |f(x,z')| Q_z(dz').
\]
Note that this integral is well defined, as for a Carath\'eodory function $f$ and compact set $K\subseteq X$, $z'\mapsto \max_{x\in K} |f_x(z')|$ is Borel measurable by the Measurable Maximum Theorem, see \cite{AB}.
Consider the space
of Carath\'eodory functions $f$ on $X\times Z$ for which  $p_ {K,z}(f)$ is finite
\[
E:=\{f\in \ca{X\times Z}\,:\,  p_{K,z}(f) <\infty \mbox{ for all $K\in \mathcal{K}$, $z\in Z$} \}.
\]
It is clear that $p_{K,z}$ is a seminorm on $E$ (but not a norm, obviously). Given $K\in \mathcal{K}$ and $z\in Z$, let
\[
E_{K,z}:=\{f\in \ca{K\times Z}\,:\, p_{K,z}(f) <\infty\}.
\]
\begin{lemma}\label{lemma1}
For each $K\in \mathcal{K}$ and $z\in Z$, $(E_{K,z}, p_{K,z})$ is a Banach space.
\end{lemma}
\begin{proof}
It is clear that $E_{K,z}$ is a vector space and that $p_{K,z}$ is a norm Note that $p_{K,z}(f)=0$ implies $\int_Z |f(x,z')|Q_z(dz') = 0$ for all $x\in K$, hence $f(x,z')=0$ for all $x\in K$, $Q_z$--a.s.. Hence $f=0$. Let $\{f_n\}$ be a Cauchy sequence. Then $p_{K,z}(f_n-f_m)\to 0$ as $n,m\to \infty$. There is $n(1)$ such that $p_{K,z}(f_{n(1)}-f_n) < 2^{-1}$ for all $n\ge n(1)$. Now choose $n(2)>n(1)$ such that  $p_{K,z}(f_{n(2)}-f_{n(1)}) < 2^{-2}$ for all $n\ge n(2)$ and so on. Thus, we obtain a subsequence $n(j)$ such that $p_{K,z}(f_{n(j+1)}-f_{n(j)}) < 2^{-j}$ for $j=1,2,\ldots$. Let, to simplify notation, $\overline f_j := f_{n(j)}$ and let $g_j(z') := \sup_{x\in K} |\overline f_{j+1}(x,z')-\overline f_{j}(x,z')|$. Note that, for all $N=1,2,\ldots$
\[
\sum_{j=1}^N \int_Z g_j(z') \mu_z(dz') < 2^{-1} + \cdots +2^{-N} <1,
\]
hence $G_N=\sum_{j=1}^N g_j$ is a monotone sequence of positive and integrable functions. By the Monotone Convergence Theorem, the function $G(z') =\lim_{N\to \infty} G_N(z')$ is integrable, and thus  finite $Q_z$--a.e., see \cite{Dieudonne}, (13.6.4). From this it follows that the series
$
\sum_{j=1}^{\infty}  g_j(z')
$ converges $Q_z$--a.e. Since
\begin{equation}\label{Mtest}
\sum_{j=1}^{\infty} |\overline f_{j+1}(x,z') - \overline f_{j}(x,z') |  \le \sum_{j=1}^N  g_j(z'),
\end{equation}
the series on the left hand side also converges for any $x\in K$, $Q_z$--a.e. Consider the series
\begin{equation}\label{Mtest2}
\overline f_{1}(x,z') +  \sum_{j=1}^{\infty} \left(\overline f_{j+1}(x,z') - \overline f_{j}(x,z')\right).
\end{equation}
It converges $Q_z$--a.e., and if $f(x,z')$ is its sum, note that by \eqref{Mtest} and the Weierstrass M--test, the convergence is uniform in the compact set $K$. Since every $\overline f_j$ is a Carath\'eodory function, the limit $f(x,z')$ is continuous in $x$. Moreover, the limit is measurable in $z'$ and $\int_Z \max_{x\in K}|f(x,z')|Q_z(dz')\le \int_Z \max_{x\in K} \left(|\overline f_1(x,z')| + G(x,z')\right) Q_z(dz')<\infty$, hence $f\in E_{K,z}$. Let us see that the convergence of $\{f_n\}$ to $f$ is in the norm $p_{K,z}$. To show this, note that the convergence of the series \eqref{Mtest2} is uniform in $x\in K$, hence
\[
\max_{x\in K} |\overline f_{j}(x,z')-f(x,z') |
\]
tends to 0 as $j\to \infty$ and thus, $p_{K,z}(f_{n(j)}-f)\to 0$ as $j\to \infty$. However, $p_{K,z}(f_{n}-f_m)\to 0$ as $n,m\to \infty$, hence
\[
\lim_{n,j\to \infty} p_{K,z}(f_{n}-f) \le  \lim_{n,j\to \infty} p_{K,z}(f_n - f_{n(j)}) +  \lim_{n,j\to \infty} p_{K,z}(f_{n(j)}-f) = 0.
\]
Thus, $\lim_{n\to \infty} p_{K,z}(f_{n}-f) = 0$.
\end{proof}

\strut
Let $\mathcal{H}$ be the family of all finite subsets $H$ of $Z$ and let, for $H\in \mathcal{H}$,
\[
p_{K,H}(f) = \sup_{z\in H} p_{K,z}.
\]
Note that $p_{K,H}$ is a norm on
\[
E_{K,H} :=\{f\in \ca{K\times Z}\,:\, p_{K,H}(f) <\infty\},
\]
and that, by Lemma \ref{lemma1}, $(E_{K,H}, p_{K,H})$ is a Banach space, since the norms $p_{K,H}$ and $p_{K,z}$ are equivalent, for any $z\in H$, and generate the same topology, see \cite{Dieudonne}, (12.14.7). The importance of choosing this family of seminorms instead of the original one is that $\mathcal{P}_{\mathcal{H}}$ is a directed family by inclusion. That is, if we define $H\le H'$ if $H\subseteq H'$, then for any $H,H'\in \mathcal{H}$, there exists $H''\in \mathcal{H}$ such that $p_{K,H''}\ge \max\{p_{K,H}, p_{K,H'}\}$.

An element in $E_{K,H}$ can be written $f+M_{K,H}$ in equivalence class notation
where
\[
M_{K,H}=\{f\in E_{K,H}\,:\, f_x = 0 \mbox{ $Q_z$--a.e., for all $x\in K$, for all $z\in H$}\}.
\]
Consider the sets
\[
\mathcal{E}_{K,H}:=\{f\in \ca{K\times Z}\,:\, p_{K,H}(f) <\infty\},
\]
and the set of Carath\'eodory functions that are integrable with respect to all $z\in Z$,
\[
\mathcal{E}_K:=\{f\in \ca{K\times Z}\,:\, p_{K,H}(f) <\infty\mbox{ for all  $H\in \mathcal{H}$}\}.
\] 
Note that $E_{K,H}$ is the quotient space $\mathcal{E}_{K,H}/M_{K,H}$.
Let us define $M_K=\bigcap_{H\in \mathcal{H}} M_{K,H}$ and consider the quotient space $E_K:=\mathcal{E}_K/M_K$, formed by equivalence classes of functions in $\mathcal{E}_K$ with respect to the relation $M_K$. That is,
two functions of $\mathcal{E}_K$ are in the same equivalence class if and only if, for any $x\in K$, $f_x(z')=g_x(z')$ $Q_z$--a.e., for all $z\in Z$.
\begin{lemma}\label{EK}
\[
E_K = {\displaystyle \bigcap _{H\in \mathcal{H}} E_{K,H}}.
\]
\end{lemma}
\begin{proof}
Let $f+M_K\in E_K$. Then $f+M_K\subseteq f+M_{K,H}$ and $p_{K,H}(f+M_K) = p_{K,H}(f+M_{K,H})  < \infty$ for all $H\in \mathcal{H}$, hence $f+M_K\in E_{K,H}$ for all $H\in \mathcal{H}$. Reciprocally, if $g$ is a representative element of an equivalence class in ${\bigcap _{H\in \mathcal{H}} E_{K,H}}$, then there is $f\in \mathcal{E}_{K,H}$ such that $g=f+M_{K,H}$ for all $H$. Hence $g-f\in M_{K,H}$ for all $H$, and thus $f-g\in M_K$, or $g=f+M_K$, and hence $g\in E_K$.
\end{proof}
We consider on $E_K$ the topology $\tau_H$ generated by the family of seminorms $\mathcal{P}_{\mathcal{H}}=(p_{K,H})_{H\in \mathcal{H}}$. 
We show in the next result that $(E_K, \mathcal{P}_{\mathcal{H}})$ is the projective limit of the family of Banach spaces $(E_{K,H})_{H\in \mathcal{H}}$, $\lim_{\leftarrow}E_{K,H}$, and thus it is a complete locally convex topological space.
\begin{lemma}\label{lemma2} For each $K\in \mathcal{K}$, $(E_K, \mathcal{P}_{\mathcal{H}})$ is a Hausdorff complete locally convex space.
\end{lemma}
\begin{proof}
Let $\tau_H$ be the projective topology on $E_K$ with respect to the family of Banach spaces $(E_{K,H}, p_{K,H})_{H\in \mathcal{H}}$. The family $\mathcal{H}$ is directed by inclusion. Given $H,H'\in \mathcal{H}$ with $H\le H'$, let the linear mapping $q_{HH'}:E_{K,H'}\rightarrow E_{K,H}$ be given by $q_{HH'}(f+M_{H'}) = f + M_{H'}$. This is well defined, since $E_{K,H'}\subseteq E_{K,H}$ and $M_{H'}\subseteq M_H$, for $H\le H'$. Clearly, each $q_{HH'}$ is continuous. Then, by Example 2.2.7 in \cite{BogachevSmolyanov}, the projective limit $\lim_{\leftarrow}E_{K,H}$ coincides with $F:=\bigcap_{H\in \mathcal{H}} E_{K,H}$. But $E_K=F$ by Lemma \ref{EK}. To see that $\tau_H$ is Hausdorff (or separated), we have to prove that for all nonzero element $g\in \lim_{\leftarrow}E_{K,H}$, there is $H\in \mathcal{H}$ and a neighborhood of the zero equivalence class, $U_{K,H}\subseteq E_{K,H}$, such that $g+M_{K,H}\notin U_{K,H}$, see \cite{Schaefer} (II.5.1). Since $g$ is nonzero, we have that $g+M_{K,H}$ is nonzero for all $H\in \mathcal{H}$, hence there is $H\in \mathcal{H}$ such that $p_{K,H}(g+M_{K,H})=p_{K,H}(g) =\delta >0$. Then, letting $U_{K,H} = \{f+M_{K,H}\in E_{K,H}\,:\, p_{K,H}(f+M_{K,H}) < \delta/2\}$, we are done.

To conclude the proof, note that the projective limit of Banach spaces is a complete locally convex space, see \cite{Schaefer} (II.5.3).
\end{proof}

\textsc{Proof of Lemma \ref{lemma:function_space}}
As in the the previous lemma, we work with the directed family of seminorms $\mathcal{P}_H$. As discussed above, it generates the same topology as $\mathcal{P}$.
Let $P = \prod_{K\in \mathcal{K}} E_K$, endowed with the Tychonoff product topology. By Lemma \ref{lemma2}, each $E_K$ is a complete locally convex space, thus $P$ is also a complete locally convex space. Moreover, $P$ is obviously Hausdorff, since the family of seminorms is separating.
Let us see that there is a linear homomorphism between $E$ and $P$, so the lemma follows. Let $\phi:E \longrightarrow P$ be defined by $\phi(f) = (f_K)_{K\in \mathcal{K}}$, where $f_K\in E_K$ is the restriction of $f$ to $K\times Z$ (we dismiss now the equivalence class notation used in the proof of Lemma \ref{lemma2}, as there is no possible confusion here, as the equivalence relation is $M$, defined prior to Lemma \ref{lemma2}). It is clear that $\phi$ is linear and one-to-one, since $f\neq0$ implies that there is $K\in \mathcal{K}$ such that $f_K\neq 0$. It is also suprajective, since, under our hypotheses on $X$, $Z$ and $Q$, every function $f_K$ in $E_K$ can be extended to a Carath\'eodory function $f^K:X\times Z\longrightarrow \R$, see \cite{Kucia}, Corollary 3. Consider the function $f(x,z)=f^K(x,z)$ if $x\in K$. This definition is consistent, since for another compact set $K'$, if $x\in K\cap K'$, then $f^{K}(x,z)=f_{K\cap K'}(x,z)=f^{K'}(x,z)$. Moreover, $f$ is a Carath\'eodory function: for each $z\in Z$, the restriction of $f$ to a compact set $K$ is continuous; hence, since $X$ is locally compact, $f$ is continuous; see \cite{Willard}, Lemma 43.10. Also, it is trivial that $f$ is Borel measurable with respect to $z$. Hence, for any  $(f_K)_{K\in \mathcal{K}}\in P$, we have proved the existence of a function $f$ in $E$ for which $\phi(f)=(f_K)_{K\in \mathcal{K}}$, and hence $\phi$ is suprajective. It remains to show that $\phi$ is continuous and that it is open. Let $\pi_K:P\rightarrow E_K$ be the projection of $E$ onto $E_K$, defined  as follows: if $f\in E$, then $\pi_K(f)=f_K$. The mappings $\pi_K$ are continuous by definition of the Tychonoff topology. Note that $\pi_K\circ \phi = \pi_K$, hence by \cite{Schaefer} (II.5.2), $\phi$ is continuous. Moreover, from the previous identity, $\pi_K\circ\phi^{-1} = \pi_K$, hence by the same argument as above, $\phi^{-1}$ is continuous.

\section{Continuity of the Markov operator}\label{App:Con}
In this section we investigate the continuity of the fixed point of the Bellman operator in the variables $(x,z)$. Our exploration is not the most general possible. We restrict ourselves to a case which is common in many models in economics.  General results about the continuity of the Markov operator $M$ can be consulted in \cite{Serfozo} and \cite{HLL}. We state the following simple result.
\begin{lemma}\label{lemma:cont}
Let $f\in \mathcal{L}^1(Z;C(X))$. Suppose that there is a $\sigma$-finite measure $\lambda$ on $Z$ such that $Q_z$ is absolutely continuous with respect to $\lambda$, for all $z\in Z$, with density (or Radon-Nicodym derivative) $\varphi(z,z')$, continuous with respect to $z$ and such that, for all compact set $K_1$ in $X$ and $K_2$ in $Z$, there exists a function $h\in L^1(Z)$ such that $|f(x,z')\varphi(z,z')|\le h(z')$ for almost all $z'\in Z$ and all $x\in K_1$, $z\in K_2$. Then $Mf$ is continuous in $(x,z)$.
\end{lemma}
\begin{proof}
The assumptions on $Q$ imply
\[
Mf(x,z) = \int_Z f(x,z') Q_z(dz') = \int_Z f(x,z') \varphi(z,z') \lambda (dz').
\]
Theorem 20.3 in \cite{AB} applies to $f(x,z') \varphi(z,z')$, hence $Mf$ is continuous.
\end{proof}

The issue of continuity of the value function in the unbounded case (and unbounded space of shocks) is not an easy one. The translation of Lemma 12.14 in \cite{SLP} to this case is not straightforward, even if the Markov chain is strong Feller continuous. Recall that $Q$ has the weak (strong) Feller property if $M$ maps bounded continuous functions (\textit{resp.} bounded measurable functions) on $Z$ into bounded continuous functions.

To see the kind of problems that may emerge for unbounded functions, consider the following example. Let $Z=[0,\infty)$ and let the transition function $Q:Z\times \mathcal{Z} \longrightarrow \R$ be defined as follows:
$
Q(0,B)=\delta_0(B)
$, where $\delta_0$ is the Dirac measure at the point 0, that is, $\delta_0(B)=1$ if $0\in B$ and $\delta_0(B)=0$ otherwise.
For $0<z<1$, $
Q(z,B)=\int_{B} dF_z(z')$, 
where
\[
F_z(z')=\left\{\begin{array}{ll} 0,&\hbox{if $z'=0$;}\\
z' z^2+1-z,&\hbox{if $0<z'\le \frac 1z$};\\
1,& \hbox{if $z'>\frac 1z$}.
\end{array}\right.,
\]
Finally, for $z\ge 1$, $Q(z,B)=\lambda(B\cap [0,1])$,
where $\lambda$ denotes the Lebesgue measure of $\R$.

Note that, for $0<z<1$, $F_z$ is a distribution function: it is nondecreasing, continuous except at 0, where the right sided limit exists, $0\le F_z\le 1$, and
\[
\int dF(z') = (z' z^2+1-z-0)|_{z'=0} + \int_0^{\frac 1z} z^2 dz'  = 1-z+z=1.
\]
Moreover, it is clear that $Q(\cdot,B)$ is Borel measurable. Thus, $Q$ is a transition function.
Let $f(y,z)=f(z)$ be independent of $y$ and continuous in $z$. Then $Mf$ is well defined in this particular example and depends only on $z$, with
$
(Mf)(0)=\int f(z')Q(0,dz') = f(0).
$
For $0<z<1$ we have
\begin{align*}
(Mf)(z) &= \int f(z') Q(z,dz')\\
&\quad  = \int f(z') dF_z(z')\\
&\quad = f(0)(z' z^2+1-z-0)|_{z'=0} + \int _0^{\frac 1z} f(z')z^2dz' \\
&\quad= f(0)(1-z) + z^2 \int _0^{\frac 1z} f(z')dz'.
\end{align*}
Note that as $f$ is continuous, the integral above exists, for any $0<z<1$ and $Mf$ is continuous for $0<z<1$.
For $z\ge 1$, $Mf$ is constant and given by
\[
(Mf)(z) = \int f(z')Q(z,dz') = \int_{[0,1]} f(z')dz'.
\]
Now, if $f$ is measurable and bounded, there is $k>0$ such that $|f|\le k$, hence $-kz^2 z^{-1} \le  \int _0^{\frac 1z} f(z')dz' \le kz^2 z^{-1}$, so $ \int _0^{\frac 1z} f(z')dz'$ tends to 0 as $z\to 0^+$, and then $(Mf)(z)$ tends to $f(0)=(Mf)(0)$. Also, $(Mf)(z)$ tends to $\int_{[0,1]} f(z')dz' = (Mf)(1)$ as $z\to 1^-$. Thus $Mf$ is continuous. Thereby, $Q$ is strong Feller continuous. However, considering the unbounded function $g(z') = z'$, we obtain $Mg(0)=g(0)=0$ and $Mg(z)=\frac 12$ for $z>0$, thus $Mg$ is discontinuous at 0.

It is not difficult to find non--trivial continuous functions $U$ for which the dynamic programming equation with transition probability $Q$ admits discontinuous solutions. For instance, let $u(z,c)=(1+z)c$ be an utility function that depends on consumption $c$ and shock $z$, and let $\Gamma(m)=[0,m+y]$, where $m\ge 0$, $y>0$ is a constant,$X=\R_+$, $Z=[0,\infty]$, and let a discount factor $\beta$ such that $\frac 32\beta <1$. The dynamic programming equation is
\[
v(m,z) = \max_{m'\in [0,m+y]} \Big\{(1+z)(m+y-m') + \beta \int_{[1,\infty)} v(m',z') Q_z(dz') \Big\},
\]
Notice that this specification corresponds to a pure currency economy model with linear utility, where agents' preferences are subject to random shocks. These random shocks are assumed to be governed by the Markov chain $Q$ described above. See \cite{SLP} for further details about this model. We are simply interested in showing that the value function is not jointly continuous in $(m ,z)$.
It is easily checked that
\[
v(m,z) = \left\{\begin{array}{ll}m+y+ y\frac {\beta}{1-\beta},&\mbox{if $z=0$;}\\
(1+z)(m+y) + \frac 32  y \frac {\beta}{1-\beta},&\mbox{if $z>0$},\end{array}
\right.
\]
is a solution in the class $\ca{\R_+\times\R_+}$, which is not continuous in $z$.

\section*{Acknowledgements}
Support from the Ministerio de Ciencia, Innovaci\'on y Universidades -- Agencia Estatal de Investigaci\'on, grants ECO2017-86261-P,
MDM 2014-0431 and Comunidad de Madrid (Spain), grant MadEco-CM S2015/HUM-3444 is gratefully acknowledged. 

\end{document}